\documentclass[hyper]{prop2015}
\usepackage[english]{babel}
\category{Proceedings}
\subtitle{\href{http://www.maths.dur.ac.uk/lms/109/index.html}{LMS/EPSRC Durham Symposium on Higher Structures in M-Theory}}
\title{Higher Structures in Algebraic Quantum Field Theory}
\author[M.~Benini]{Marco Benini\inst{a}}
\author[A.~Schenkel]{Alexander Schenkel\inst{b,}\footnote{Corresponding author email: \href{mailto:alexander.schenkel@nottingham.ac.uk}{\textsf{alexander.schenkel@notting\-ham.ac.uk}}}}
\address[1]{Fachbereich Mathematik, Universit\"at Hamburg, Bundesstr.~55, 20146 Hamburg, Germany}
\address[2]{School of Mathematical Sciences, University of Nottingham, University Park, Nottingham NG7 2RD, United Kingdom}
\shortauthors{M.~Benini and A.~Schenkel}
\begin{acknowledgements}
We would like to thank our collaborators S.~Bruinsma, M.~Perin,
U.~Schreiber, R.~Szabo and L.~Woike for their work
and contributions to the {\em homotopical algebraic 
quantum field theory program} that we initiated some years ago.
A.S.~also would like to thank the organizers and sponsors of the LMS-EPSRC 
Durham Symposium {\em Higher Structures in M-Theory}
for this very interesting and fruitful conference.
The work of M.B.\ is supported by a research grant funded by 
the Deutsche Forschungsgemeinschaft (DFG, Germany). 
A.S.\ gratefully acknowledges the financial support of 
the Royal Society (UK) through a Royal Society University 
Research Fellowship, a Research Grant and an Enhancement Award. 
\end{acknowledgements}
\begin{abstract}
A brief overview of the recent developments of operadic and higher categorical techniques in algebraic quantum field theory is given. The relevance of such mathematical structures for the description of gauge theories is discussed.
\end{abstract}
\shortabstract

\usepackage{amsthm}
\usepackage{bbm}
\usepackage[all,cmtip,2cell]{xy}
\usepackage{tikz}
\usetikzlibrary{shapes.geometric}
\usetikzlibrary{cd}

\theoremstyle{plain}
\newtheorem{theo}{Theorem}[section]
\newtheorem{lem}[theo]{Lemma}
\newtheorem{propo}[theo]{Proposition}
\newtheorem{cor}[theo]{Corollary}

\theoremstyle{definition}
\newtheorem{defi}[theo]{Definition}

\newtheorem{ex}[theo]{Example}
  
\newtheorem{rem}[theo]{Remark}

\newtheorem{var}[theo]{Variation}
\newtheorem{problem}[theo]{Open Problem}

\def\Loc{\mathbf{Loc}}
\def\Locc{\mathbf{Loc}_{\text{\large $\diamond$}}^{}}
\def\COpen{\mathbf{COpen}}
\def\Open{\mathbf{Open}}
\def\CastAlg{C^\ast\mathbf{Alg}}
\def\astAlg{{}^{\ast}\mathbf{Alg}}
\def\Alg{\mathbf{Alg}}
\def\QFT{\mathbf{QFT}}
\def\OCat{\mathbf{OrthCat}}
\def\Op{\mathbf{Op}}
\def\Set{\mathbf{Set}}
\def\CC{\mathbf{C}}
\def\DD{\mathbf{D}}
\def\MM{\mathbf{M}}
\def\HH{\mathbf{H}}
\def\Grpd{\mathbf{Grpd}}
\def\sSet{\mathbf{sSet}}
\def\Cart{\mathbf{Cart}}
\def\Man{\mathbf{Man}}
\def\Disk{\mathbf{Disk}}
\def\Ch{\mathbf{Ch}}

\def\AAA{\mathfrak{A}}
\def\BBB{\mathfrak{B}}
\def\CCC{\mathfrak{C}}

\def\bbC{\mathbb{C}}
\def\bbR{\mathbb{R}}
\def\bbL{\mathbb{L}}
\def\bbS{\mathbb{S}}

\def\O{\mathcal{O}}
\def\P{\mathcal{P}}
\def\E{\mathcal{E}}

\def\As{\mathsf{As}}

\def\oone{\mathbbm{1}}
\def\id{\mathrm{id}}

\def\Hom{\mathrm{Hom}}
\def\holim{\mathrm{holim}}
\def\colim{\mathrm{colim}}
\def\hocolim{\mathrm{hocolim}}
\def\Lan{\mathrm{Lan}}

\def\op{\mathrm{op}}
\def\dd{\mathrm{d}}

\newcommand\ovr[1]{\overline{#1}}
\newcommand\und[1]{\underline{#1}}

\begin{document}
\maketitle


\section{\label{sec:background}Background on AQFT}
Algebraic quantum field theory (AQFT) is a mathematical
framework to formalize and investigate quantum field theories (QFTs) on 
{\em Lorentzian} manifolds, i.e.\ on space-times in the sense of general relativity. 
We have emphasized the adjective Lorentzian because this is what makes AQFT different
from other mathematical\linebreak approaches to QFT,
such as (extended) topological QFT  \cite{Atiyah:1989vu,Lurie:2009keu}
or the factorization algebra approach of Costello and Gwilliam \cite{costello2016factorization}.
The original framework of Haag and Kastler \cite{Haag:1963dh} was restricted
to QFTs on Minkowski space-time, but a more flexible version of AQFT
that works on {\em all} (globally hyperbolic) Lorentzian manifolds was later developed
by Brunetti, Fredenhagen and Verch \cite{Brunetti:2001dx}. AQFT turns out to be a very powerful and 
successful framework not only for proving model-independent results for QFTs,
but also for studying concrete applications with a high level of mathematical rigor.
We refer to \cite{Brunetti:2015vmh} for an overview of some of the recent advances in AQFT.

Before we can provide a definition of what an AQFT is,
we have to make precise on which space-times
we would like our QFTs to live. 
We refer to \cite{Baer:0806.1036} for a concise introduction to Lorentzian geometry.
In order to avoid pathologies, one typically considers
only globally hyperbolic Lorentzian manifolds. These are
Lorentzian manifolds $M$ for which there exists a Cauchy
surface $\Sigma\subset M$, i.e.\ a codimension $1$ hypersurface 
that is intersected precisely once by every inextensible causal curve.
We further would like that $M$ is oriented and time-oriented
in order to have a volume form and a way to distinguish
between future and past. We collect all oriented and time-oriented
globally hyperbolic Lorentzian manifolds (of a fixed dimension $m\geq 2$)
in a category that we denote by $\Loc$. The morphisms $f:M\to N$ in $\Loc$
are orientation and time-orientation preserving isometric embeddings 
of $M$ into $N$ such that the image $f(M)\subseteq N$ is open and causally convex,
i.e.\ every causal curve in $N$ that starts and ends in $f(M)$ is entirely contained in $f(M)$.
There exists a distinguished class $W\subseteq \mathrm{Mor}\Loc$ 
of $\Loc$-morphisms, called Cauchy morphisms, 
which is given by all $f:M\to N$ such that the image $f(M)\subseteq N$ 
contains a Cauchy surface of $N$. Loosely speaking, one should
think of $f(M)\subseteq N$ as a `time slab' with respect to a time coordinate on $N$.
The category $\Loc$ may be endowed with a further structure that encodes  causal
independence of subspace-times. We call a pair of $\Loc$-morphisms 
$(f_1:M_1\to N,f_2:M_2\to N)$ to a common target 
causally disjoint if their images $f_1(M_1)\subseteq N$ and $f_2(M_2)\subseteq N$
are causally disjoint subsets of $N$, i.e.\ there exists no causal curve connecting $f_1(M_1)$ and $f_2(M_2)$.
The collection of all causally disjoint pairs of $\Loc$-morphisms
is denoted by ${\perp_\Loc} \subseteq \mathrm{Mor}\Loc \, {}_{\mathrm{t}}{\times}_{\mathrm{t}} \, \mathrm{Mor}\Loc$
and we shall simply write $f_1\perp_{\Loc}f_2$ whenever $(f_1,f_2)\in{\perp_{\Loc}}$.
The following variant of AQFT was proposed in \cite{Brunetti:2001dx} and it is called 
locally covariant QFT.
\begin{defi}\label{def:LCQFT}
A {\em locally covariant QFT} is a functor
$\AAA : \Loc\to \CastAlg$ to the category of unital 
$C^\ast$-algebras that satisfies the following properties:
\begin{enumerate}[i)]
\item {\em Isotony:} For all $\Loc$-morphisms $f:M\to N$, the $\ast$-homomorphism 
$\AAA(f) : \AAA(M)\to \AAA(N)$ is injective.

\item {\em Einstein causality:} For all causally disjoint $f_1\perp_{\Loc}f_2$, the induced commutator
\begin{flalign}
\big[\AAA(f_1)(a), \AAA(f_2)(b)\big]_{\AAA(N)} =0
\end{flalign}
is zero, for all $a\in \AAA(M_1)$ and $b\in \AAA(M_2)$.

\item {\em Time-slice axiom:} For all Cauchy morphisms $f\in W$, the $\ast$-homomorphism
$\AAA(f):\AAA(M)\to\AAA(N)$ is an isomorphism.
\end{enumerate}
\end{defi}

The physical interpretation is as follows: The $C^\ast$-algebra $\AAA(M)\in\CastAlg$
associated to $M\in\Loc$ describes the quantum observables of the theory that one
can measure in the space-time $M$. The $\ast$-homomorphism $\AAA(f) : \AAA(M)\to \AAA(N)$
associated to a $\Loc$-morphism $f:M\to N$ pushes forward
observables along this space-time embedding. The isotony axiom then says that no observables
are lost under such pushforwards, i.e.\ larger space-times are not allowed to have less observables
than smaller ones. The Einstein causality axiom formalizes that spacelike separated observables
commute with each other. Finally, the time-slice axiom implements a dynamical
law or `time evolution' as it states that the observable algebra of a small region containing 
a Cauchy surface is isomorphic to the observable algebra of the full space-time.

It might be a bit surprising for some readers that we did not mention Hilbert spaces
in Definition \ref{def:LCQFT}, which are the predominant structures 
in other approaches to QFT. The reason is that AQFT splits the 
problem of describing a QFT in two separate steps: 
First, one constructs a theory in the sense of Definition \ref{def:LCQFT},
which describes the quantum observables as abstract $C^\ast$-algebras. Second,
one studies (algebraic) states on these algebras, i.e.\ linear functionals $\omega_M : \AAA(M)\to \bbC$
that are positive $\omega_M(a^\ast a)\geq 0$, for all $a\in\AAA(M)$, and normalized  $\omega_M(\oone)=1$.
These states then define Hilbert spaces via the GNS-construction. This means that, while retaining a prominent role, in AQFT Hilbert spaces 
enter the game only at a later stage, when one analyzes the representation theory of a specific model. 
The advantage of axiomatizing only the observable algebras of a QFT,
as it was done in Definition \ref{def:LCQFT}, is that one does not have to make
any a priori choice of a distinguished `vacuum' state, but 
one treats all possible states on an equal footing. 
The reason why it is particularly important to do so is two-fold: 
1.)~different states can induce inequivalent Hilbert space representations, 
2.)~for QFTs on curved space-times there is no distinguished choice 
of vacuum state due to the lack of space-time symmetries.
Another more technical advantage is that
the observable algebras of a QFT behave local (cf.\ Einstein
causality in Definition \ref{def:LCQFT}), while states capture
non-local quantum features such as entanglement. Hence,
one can employ powerful local techniques for constructing and analyzing 
examples of AQFTs, which is particularly useful when discussing perturbatively
interacting models and their renormalization, see e.g.\ \cite{rejzner2016perturbative} for a recent overview.

We would like to add some comments about variations of
Definition \ref{def:LCQFT} that are considered 
in the literature. We shall also explain why we think such 
variations are reasonable for certain purposes.
\begin{var}
Instead of the category $\CastAlg$ of $C^\ast$-alge\-bras, one may consider also 
other categories to describe the observables of a QFT.
For example, one may consider the category $\astAlg$ of $\ast$-algebras
or various categories of topological $\ast$-algebras, sometimes over the ring
of formal power series $\bbC[[\hbar]]$. Such choices are useful for
formalizing perturbatively interacting AQFTs, where one does not 
have $C^\ast$-norms.
\end{var}

\begin{var}\label{var:Locc}
Instead of the category $\Loc$ of all oriented and time-oriented
globally hyperbolic space-times, one may consider also other categories
of space-times. For example, one may take the full subcategory $\Locc\subseteq \Loc$
of all space-times $M$ whose underlying manifold is diffeomorphic to $\bbR^m$. Note that
our notions of causal disjointness $\perp_{\Loc}$ and Cauchy morphisms $W$
restrict to this subcategory, hence we can make sense of all axioms listed
in Definition \ref{def:LCQFT} for functors $\AAA : \Locc\to\CastAlg$ defined on 
this subcategory. Physically, one interprets such AQFTs as QFTs
that are only defined on topologically trivial space-times. As another example,
one may take the category $\COpen(M)$ of causally convex open subsets $U\subseteq M$
of a fixed space-time $M$. There exists an evident functor $\COpen(M)\to \Loc$
along which we can pull back our notion of causal disjointness $\perp_{\Loc}$ and the 
Cauchy morphisms $W$. AQFTs on $\COpen(M)$ describe QFTs that
are defined on suitable subsets of a fixed space-time $M$, which is in the spirit of the
original Haag-Kastler approach \cite{Haag:1963dh}. 
\end{var}

\begin{var}\label{var:isotony}
The isotony axiom turns out to be too restrictive to capture 
important examples of QFTs that are sensitive to
topological data. For example, the functor $\AAA : \Loc\to\CastAlg$
describing gauge-invariant observables of Abelian Yang-Mills theory
violates the isotony axiom because of electric and magnetic charges
living in certain cohomology groups of the space-times. This feature was
observed first in \cite{Dappiaggi:2011zs} and it was later refined and generalized in
\cite{Benini:2013tra,Benini:2013ita,Becker:2014tla,Becker:2015ybl}.
Our current practice is to drop the isotony axiom from the definition of AQFTs.
In ongoing works, we attempt to find a suitable replacement by a 
descent (i.e.\ local-to-global) condition.
\end{var}

We would like to comment in more detail on the idea of
introducing a descent condition for AQFT. Loosely speaking,
descent means that the observable algebra $\AAA(M)$
on a `complicated' space-time $M$ can be obtained by patching together the
observable algebras $\AAA(U_i)$ on a suitable family of `simple' subspace-times
$U_i\subseteq M$. In practice, descent would allow us to replace
questions about the complicated observable algebra $\AAA(M)$ by a family of simpler questions
about the algebras $\AAA(U_i)$ and their interplay via embeddings. 
Such properties are prevalent in mathematics,
with concrete manifestations given by sheaf or cosheaf conditions.
To our surprise, it seems that descent in AQFT has not yet been
studied systematically, at least according to our best knowledge.
The reason for this might be that the most desirable 
kind of descent condition for AQFT, namely a cosheaf condition, 
turns out to be very restrictive and in particular it fails even
in the simplest examples of non-interacting AQFTs. Recall that
the cosheaf condition for the functor $\AAA : \Loc\to \CastAlg$ underlying
an AQFT states that, for every (suitable) cover $\{U_i\subseteq M\}$ by 
causally convex open subsets, the canonical morphism
\begin{flalign}\label{eqn:cosheafAQFT}
\colim\Bigg(
\!\!\xymatrix@C=1em{
\coprod\limits_{ij} \AAA(U_{ij}) \ar@<0.5ex>[r]\ar@<-0.5ex>[r]&
 \coprod\limits_{i} \AAA(U_i) }\!\!
\Bigg) ~\stackrel{\cong}{\longrightarrow}~ \AAA(M)
\end{flalign}
is an isomorphism of $C^\ast$-algebras. Here $U_{ij} := U_i\cap U_j$ denotes the intersections,
$\coprod$ the coproduct in $\CastAlg$ and $\colim$ the colimit in $\CastAlg$.
Looking at a simple example, we shall illustrate in Appendix \ref{app:cosheaf} that 
it is very hard to find covers $\{U_i\subseteq M\}$ such that this condition holds true.

A weaker but still useful descent condition can be obtained by
using ideas related to Fredenhagen's universal algebra \cite{Fredenhagen:1989pw,Fredenhagen:1993tx,Fredenhagen:1992yz}.
We shall provide here only a rough sketch of the idea and refer to 
Section \ref{subsec:universal} for a precise implementation.
Let us denote by $j: \Locc\to\Loc$ the full subcategory embedding
from Variation \ref{var:Locc}. We may restrict any AQFT $\AAA: \Loc\to \CastAlg$ to a functor 
$j^\ast\AAA := \AAA\circ j : \Locc\to\CastAlg$ on $\Locc$. It is easy to see
that this functor satisfies the analogues for $\Locc$ of the axioms in Definition \ref{def:LCQFT}, 
hence it is an AQFT that is however only defined on space-times
with underlying manifold diffeomorphic to $\bbR^m$. Applying Fredenhagen's
universal algebra construction to $j^\ast\AAA$, i.e.\ forming the left Kan extension
along the embedding functor $j : \Locc\to\Loc$ (cf.\ \cite{lang2014universal}), defines another functor
$\Lan_j j^\ast \AAA : \Loc\to \CastAlg$ on all of $\Loc$.
Note that there exists a canonical comparison natural transformation
$\epsilon_\AAA : \Lan_j j^\ast \AAA  \to \AAA$ given by the counit of the 
adjunction $\Lan_j \dashv j^\ast$. We hence may formalize
a descent condition by demanding that $\epsilon_{\AAA}$ is a 
natural isomorphism. In other words, this descent condition
formalizes the idea that the AQFT $\AAA$ on $\Loc$ is completely
determined by its values on the category $\Locc\subseteq \Loc$
of space-times diffeomorphic to $\bbR^m$.
This is similar to the descent condition in factorization 
homology \cite{Ayala:1206.5522}, which is a topological variant of
factorization algebras \cite{costello2016factorization}.

Let us mention the following issue with the latter descent condition,
which will be addressed and solved in Section \ref{subsec:universal} once
we have a more powerful mathematical machinery available.
Notice that there is no reason why the functor $\Lan_j j^\ast \AAA : \Loc\to\CastAlg$ 
should satisfy the axioms from Definition \ref{def:LCQFT}, i.e.\ $\Lan_j j^\ast \AAA$ 
is not necessarily an AQFT in the sense of Definition \ref{def:LCQFT}, even
if we drop the isotony axiom as in Variation \ref{var:isotony}.
This has as a consequence that not even simple examples of AQFTs,
such as the free Klein-Gordon theory, satisfy the present version
of the descent condition. (They however satisfy a slightly weaker
descent condition obtained by replacing $\Loc$ by the full subcategory 
$\Loc_0^{}\subseteq \Loc$ of connected space-times, cf.\ \cite{lang2014universal}.) 
A solution to this problem is to use a more refined version of
Fredenhagen's universal algebra construction that is obtained 
by methods from operad theory. This will be explained in
Section \ref{subsec:universal}.


\section{\label{sec:operads}AQFT from an algebraic perspective}
The aim of this section is to identify a colored operad
that controls the algebraic structures underlying AQFT.
The main advantages of this operadic perspective
are as follows: 1.)~It provides a suitable framework
for studying universal constructions for AQFTs, for example 
via operadic left Kan extensions. 
This will in particular allow us to formulate a precise version of the descent
condition sketched in Section \ref{sec:background}.
2.)~It provides a suitable starting point for investigating higher
structures in AQFT by importing ideas and techniques from
the homotopy theory of algebras over operads, see e.g.\ \cite{Hinich:9702015,Hinich:1311.4130}.
The second point will be discussed in detail in Section \ref{sec:homotopy}.
For details on the material presented below we refer to \cite{Benini:2017fnn}.
A generalization to other types of field theories (e.g.\ classical, linear, etc.)
can be found in \cite{Bruinsma:2018knq} and in Bruinsma's contribution to 
these proceedings \cite{contrib:bruinsma}.


\subsection{Orthogonal categories and AQFTs}
For the rest of this paper we shall adopt a very broad and 
flexible definition of AQFTs in which the space-time category
$\Loc$ is generalized to a so-called orthogonal category \cite{Benini:2017fnn}.
This allows us to treat different flavors of AQFTs, such
as locally covariant QFTs, AQFTs on a fixed space-time and
chiral conformal AQFTs, on an equal footing. 
\begin{defi}
An {\em orthogonal category} is a pair $\ovr{\CC} := (\CC,\perp)$
consisting of a small category $\CC$ and a subset
${\perp} \subseteq \mathrm{Mor}\CC \,{}_{\mathrm{t}}{\times}_{\mathrm{t}}\, \mathrm{Mor}\CC$ 
of the set of pairs of morphisms with a common target (called {\em orthogonality relation}), 
such that the following conditions hold true:
\begin{enumerate}[i)]
\item {\em Symmetry:} If $(f_1,f_2)\in{\perp}$, then $(f_2,f_1)\in {\perp}$.
\item {\em $\circ$-Stability:} If $(f_1,f_2)\in{\perp}$, then
$(g\,f_1\,h_1, g\,f_2\,h_2)\in{\perp}$, for all composable $\CC$-morphisms
$g$, $h_1$ and $h_2$.
\end{enumerate}
We denote orthogonal pairs $(f_1,f_2)\in{\perp}$ also by $f_1\perp f_2$.
An {\em orthogonal functor} $F : \ovr{\CC}\to \ovr{\DD}$ is a functor
$F : \CC\to \DD$ such that $F(f_1)\perp_\DD F(f_2)$ for all $f_1\perp_\CC f_2$.
We denote by $\OCat$ the category of orthogonal categories 
and orthogonal functors.
\end{defi}
\begin{ex}\label{ex:LocLocc}
The pair $\ovr{\Loc} = (\Loc,\perp_\Loc)$ discussed in Section \ref{sec:background}
is an orthogonal category. Endowing the subcategories $j : \Locc\to \Loc$
and $j_M : \COpen(M)\to \Loc$ from Variation \ref{var:Locc} with the pullback
orthogonality relations $j^\ast(\perp_\Loc)$ and $j_M^\ast(\perp_\Loc)$
defines orthogonal categories $\ovr{\Locc} $ and $\ovr{\COpen(M)}$. The embedding
functors $j : \ovr{\Locc} \to \ovr{\Loc}$ and $j_M : \ovr{\COpen(M)} \to \ovr{\Loc}$
are orthogonal functors. For an example that is not directly related to Lorentzian geometry, 
consider the category $\mathbf{Int}(\mathbb{S}^1)$ of open intervals $I\subset \mathbb{S}^1$ 
of the circle with morphisms given by subset inclusions $I\subseteq J \subset \mathbb{S}^1$.
We define an orthogonality relation on $\mathbf{Int}(\mathbb{S}^1)$ by declaring
two morphisms $I_1,I_2 \subseteq J\subset \mathbb{S}^1$  to be orthogonal if and only if
$I_1\cap I_2 = \emptyset$ are disjoint intervals. The corresponding 
orthogonal category $\ovr{\mathbf{Int}(\mathbb{S}^1)}$ features in chiral conformal 
QFT, see e.g.\ \cite{Kawahigashi:2015zha}.
\end{ex}

Throughout the whole section we shall fix a closed symmetric monoidal category
$\MM = (\MM,\otimes, I)$, which we further assume to be bicomplete,
i.e.\ all small limits and colimits exist in $\MM$.
\begin{defi}\label{def:QFTcats}
Let $\ovr{\CC} = (\CC,\perp)$ be an orthogonal category.
An {\em $\MM$-valued AQFT on $\ovr{\CC}$} is a functor
$\AAA : \CC\to \Alg_{\As}(\MM)$ to the category of
associative and unital algebras in $\MM$ that satisfies
the $\perp$-commutativity property: For all
$(f_1: c_1\to c) \perp (f_2: c_2\to c)$, the diagram
\begin{flalign}
\xymatrix@C=6em{
\ar[d]_-{\AAA(f_1)\otimes \AAA(f_2)}\AAA(c_1)\otimes \AAA(c_2)\ar[r]^-{\AAA(f_1)\otimes \AAA(f_2) } & \AAA(c)^{\otimes 2}\ar[d]^-{\mu_c}\\
\AAA(c)^{\otimes 2} \ar[r]_-{\mu_c^{\op}} & \AAA(c)
}
\end{flalign}
in $\MM$ commutes, where $\mu_c^{(\op)}$ denotes the (opposite) multiplication
in the algebra $\AAA(c)$. The category of
$\MM$-valued AQFTs on $\ovr{\CC}$ is defined as the full subcategory
\begin{flalign}
\QFT(\ovr{\CC}) \subseteq \Alg_{\As}(\MM)^{\CC}
\end{flalign}
of the functor category that consists of all $\perp$-commutative functors.
\end{defi}

\begin{rem}\label{rem:AQFTdefinition}
The following remarks are in order:
\begin{enumerate}[i)]
\item Motivated by Variation \ref{var:isotony}, we decided to omit the
isotony axiom in Definition \ref{def:QFTcats} because it is often 
violated in examples. 
\item Comparing Definitions \ref{def:QFTcats} and \ref{def:LCQFT},
it seems at first sight that we neglect the time-slice axiom in 
Definition \ref{def:QFTcats}. This is however not the case.
We shall prove in Proposition \ref{prop:timeslicelocalization}
below that the time-slice axiom may be encoded by localizing
the orthogonal category $\ovr{\Loc}$ at all Cauchy morphisms 
$W\subseteq\mathrm{Mor} \Loc$, which defines another orthogonal
category $\ovr{\Loc[W^{-1}]}$. 
\item Note that Definition \ref{def:QFTcats}
does not refer explicitly to $\ast$-involutions on algebras. These can be included
in a relatively straightforward way by choosing as target category $\MM$ an {\em involutive} closed symmetric 
monoidal category, see \cite{Benini:2018mcq} for the technical details. 
\end{enumerate}
\end{rem}

\begin{problem}
Coming back to the last point of the previous remark,
we would like to emphasize that, even though it is clear how
to include $\ast$-algebras in Definition \ref{def:QFTcats}, capturing $C^\ast$-algebras
as in Definition \ref{def:LCQFT} is more subtle and still an open problem.
The reason is that the category $\CastAlg$ of $C^\ast$-algebras is not 
(or at least not known to be) a category of $\ast$-algebras in a suitable 
involutive closed symmetric monoidal category. There are proposals
in the literature \cite{egger2011involutive,Egger:2007aa} to replace $\CastAlg$ by the category
of $\ast$-algebras in the involutive closed symmetric monoidal category
of {\em operator spaces} in order to obtain a categorical approach
to the theory of operator algebras.
To the best of our knowledge, such an approach has not been applied to AQFT yet.
\end{problem}

As promised in Remark \ref{rem:AQFTdefinition}, we shall now prove that
Definition \ref{def:QFTcats} includes the case of AQFTs satisfying the time-slice axiom.
Let $\ovr{\CC} = (\CC,\perp)$ be any orthogonal category and $W\subseteq\mathrm{Mor} \CC$
any subset of the set of morphisms. We denote by $\CC[W^{-1}]$
the localization of the category $\CC$ at $W$ and by $L:\CC\to\CC[W^{-1}]$
the corresponding localization functor.
(We refer to  \cite[Section 7.1]{Kashiwara:2006ww} for details on localizations of categories.)
We define $\perp_W$ to be the
smallest orthogonality relation on $\CC[W^{-1}]$ such that $L(f_1)\perp_W L(f_2) $,
for all $f_1\perp f_2$. This implies that $L : \ovr{\CC}\to \ovr{\CC[W^{-1}]}$ is 
an orthogonal functor. We shall denote by
\begin{flalign}
\QFT(\ovr{\CC})^{W-\mathrm{const}}\subseteq \QFT(\ovr{\CC})
\end{flalign}
the full subcategory of AQFTs satisfying the {\em $W$-constancy property}, i.e.\
$\AAA(f) :\AAA(c)\to\AAA(c^\prime)$ is an isomorphism in $ \Alg_{\As}(\MM)$, 
for all $(f:c\to c^\prime)\in W$.
\begin{propo}\label{prop:timeslicelocalization}
The pullback functor $L^\ast := (-)\circ L : \Alg_{\As}(\MM)^{\CC[W^{-1}]} \to \Alg_{\As}(\MM)^\CC$
on functor categories restricts to an equivalence of categories 
\begin{flalign}\label{eqn:localizationtmp}
L^\ast : \QFT(\ovr{\CC[W^{-1}]})\stackrel{\sim}{\longrightarrow} \QFT(\ovr{\CC})^{W-\mathrm{const}}\quad.
\end{flalign}
\end{propo}
\begin{proof}
One immediately observes that, for every 
$\BBB\in \QFT(\ovr{\CC[W^{-1}]}) \subseteq  \Alg_{\As}(\MM)^{\CC[W^{-1}]}$, the 
functor $L^\ast \BBB =\BBB\circ L: \CC\to \Alg_\As(\MM)$ is $W$-constant
and $\perp$-commutative. Hence, $L^\ast$ restricts
to \eqref{eqn:localizationtmp}. By definition of localization, $L^\ast$ is fully faithful
and hence so is its restriction \eqref{eqn:localizationtmp} to full subcategories. 
It remains to show that \eqref{eqn:localizationtmp} is essentially surjective. 
Given any $\AAA \in  \QFT(\ovr{\CC})^{W-\mathrm{const}}$,
there exists by definition of localization a functor $\BBB : \CC[W^{-1}]\to \Alg_\As(\MM)$
and a natural isomorphism $\AAA \cong L^\ast \BBB$. One easily checks 
that $\BBB$ is $\perp_W$-commutative,
i.e.\ $\BBB \in \QFT(\ovr{\CC[W^{-1}]})$, by using that the orthogonality relation $\perp_W$ 
is generated by $L(f_1) \perp_W L(f_2) $, for all $f_1\perp f_2$.
\end{proof}
\begin{ex}
Applying this general proposition to the orthogonal localization $L : \ovr{\Loc}\to \ovr{\Loc[W^{-1}]}$ 
of the usual space-time category $\ovr{\Loc}$ at all Cauchy morphisms 
$W$ shows that $\QFT(\ovr{\Loc[W^{-1}]})$ is equivalent to
the category of AQFTs satisfying also the time-slice axiom.
\end{ex}


\subsection{\label{subsec:AQFToperad}AQFT operads}
From a category theoretical perspective, our Definition \ref{def:QFTcats} 
of the category of AQFTs is neither elegant nor very effective because we 
select a certain full subcategory of a functor category by demanding
additional {\em \underline{properties}}.
For example, it is a priori unclear if the category $\QFT(\ovr{\CC})$ of AQFTs 
admits all limits and colimits. Another problem, which is related to 
our discussion of descent in Section \ref{sec:background},
is as follows: Given an orthogonal functor $F : \ovr{\CC}\to \ovr{\DD}$,
it is a priori unclear if its left Kan extension $\Lan_F : \Alg_{\As}(\MM)^\CC\to \Alg_{\As}(\MM)^\DD$ 
restricts to a functor $\QFT(\ovr{\CC})\to \QFT(\ovr{\DD})$ between
the corresponding AQFT categories. More concretely, it is a priori unclear
if constructions like Fredenhagen's universal algebra define $\perp$-commutative
functors, i.e.\ bona fide AQFTs according to Definition \ref{def:QFTcats}.
These problems were addressed and solved in \cite{Benini:2017fnn}, where we 
have shown that there exists a colored operad $\O_{\ovr{\CC}}\in \Op(\Set)$ whose 
category of algebras is the AQFT category $\QFT(\ovr{\CC})$. In this approach
$\perp$-commutativity is not formulated as a property, 
instead it is encoded as a {\em\underline{structure}}
into the operad $\O_{\ovr{\CC}}$. 

Let us recall that a colored operad $\O\in\Op(\Set)$ with values in the category of sets $\Set$ is 
a generalization of the concept of a category where morphisms  are allowed to have
more than one input. 
The following picture visualizes this basic idea:
\begin{flalign}\label{eqn:operadpic}
\text{\begin{tabular}{cc}
~~~Category ($1$-to-$1$):~~~& Colored operad ($n$-to-$1$): \\
\begin{tikzpicture}
\draw[fill=gray!10] (0,0) circle (0.2cm);
\draw (0,0) node {{\footnotesize $f$}};
\draw[thick] (0,0.2) -- (0,0.5);
\draw[thick] (0,-0.2) -- (0,-0.5);
\draw (0,0.7) node {{\footnotesize $c^\prime$}};
\draw (0,-0.65) node {{\footnotesize $c$}};
\end{tikzpicture} & \begin{tikzpicture}
\draw[fill=gray!10] (0,0.2) -- (0.5,-0.2) -- (-0.5,-0.2) -- (0,0.2);
\draw (0,0) node {{\footnotesize $o$}};
\draw[thick] (0,0.2) -- (0,0.5);
\draw (0,0.7) node {{\footnotesize $c^\prime$}};
\draw[thick] (-0.4,-0.2) -- (-0.4,-0.5);
\draw (-0.4,-0.65) node {{\footnotesize $c_1$}};
\draw[thick] (0.4,-0.2) -- (0.4,-0.5);
\draw (0.4,-0.65) node {{\footnotesize $c_n$}};
\draw (0,-0.5) node {{\footnotesize $\cdots$}};
\end{tikzpicture}
\end{tabular}}
\end{flalign}
More precisely, a {\em colored operad} $\O\in\Op(\Set)$ 
is described by the following data:
\begin{enumerate}[i)]
\item an underlying set of objects, called colors in operad theory;
\item for each tuple $(\und{c},t) = ((c_1,\dots,c_n),t)$ of colors a set
$\O\big(\substack{t \\\und{c}}\big)\in\Set$ of operations from $\und{c}$ to $t$;
\item composition maps $\gamma : \O\big(\substack{t \\ \und{c}}\big)\times\prod_{i=1}^n 
\O\big(\substack{c_i\\ \und{b}_i}\big) \to \O\big(\substack{t\\ (\und{b}_1,\dots,\und{b}_n)}\big)$;
\item unit elements  $\oone \in \O\big(\substack{t\\ t}\big)$;
\item permutation group actions $\O(\sigma) : \O\big(\substack{t \\ \und{c}}\big)\to 
\O\big(\substack{t \\ \und{c}\sigma}\big)$, where $\sigma\in\Sigma_n$ is 
a permutation of $n$ letters.
\end{enumerate}
The composition maps are assumed to be associative in the obvious sense,
unital with respect to the unit operations and also equivariant with respect
to the permutation group actions. We refer to e.g.\ \cite{yau2016colored} for a detailed
definition of colored operads.

The {\em AQFT operad} $\O_{\ovr{\CC}}\in \Op(\Set)$ for an orthogonal category
$\ovr{\CC}$  admits the following simple presentation by
generators and relations \cite{Benini:2017fnn}:
\begin{enumerate}[i)]
\item {\em Generators:}  
\begin{flalign}
\text{\begin{tabular}{c}
\begin{tikzpicture}[cir/.style={circle,draw=black,inner sep=0pt,minimum size=2mm},
        poin/.style={circle, inner sep=0pt,minimum size=0mm},scale=0.8, every node/.style={scale=0.8}]
%
%
\node[poin] (Hout) [label=above:{\small $c^\prime$}] at (0,1) {};
\node[poin] (Hin) [label=below:{\small $c$}] at (0,0) {};
\draw[thick] (Hin) -- (Hout) node[midway,left] {{\small $f$}};
%
%
\node[poin] (Uout) [label=above:{\small $c$}] at (2,1) {};
\node[poin] (Uin) [label=below:{\small $\emptyset$}] at (2,0) {};
\draw[thick] (Uin) -- (Uout) node[midway,left] {{\small $1_c$}};
\draw[thick,fill=white] (Uin) circle (0.6mm);
%
%
\node[poin] (Mout) [label=above:{\small $c$}] at (4,1) {};
\node[poin] (Min1) [label=below:{\small $c$}] at (3.7,0) {};
\node[poin] (Min2) [label=below:{\small $c$}] at (4.3,0) {};
\node[poin] (V)  [label=right:{\small $\mu_c$}] at (4,0.5) {};
\draw[thick] (Min1) -- (V);
\draw[thick] (Min2) -- (V);
\draw[thick] (V) -- (Mout);
\end{tikzpicture}
\end{tabular}}
\end{flalign}
for all $\CC$-morphisms $f\in\mathrm{Mor}\CC$ and all objects $c\in\CC$.
The first generator describes the pushforward of observables along
the space-time embedding $f:c\to c^\prime$ and the other two generators the 
unit and multiplication of observables in the space-time $c$.

\item {\em Functoriality relations:}
\begin{flalign}
\text{\begin{tabular}{c}
\begin{tikzpicture}[cir/.style={circle,draw=black,inner sep=0pt,minimum size=2mm},
        poin/.style={circle, inner sep=0pt,minimum size=0mm},scale=0.8, every node/.style={scale=0.8}]
%
%
\node[poin] (Hout) [label=above:{\small $c$}] at (0,1) {};
\node[poin] (Hin) [label=below:{\small $c$}] at (0,0) {};
\draw[thick] (Hin) -- (Hout) node[midway,left] {{\small $\oone$}};
\node[poin] (EqH) at (0.5,0.5) {$=$};
\node[poin] (HHout) [label=above:{\small $c$}] at (1.5,1) {};
\node[poin] (HHin) [label=below:{\small $c$}] at (1.5,0) {};
\draw[thick] (HHin) -- (HHout) node[midway,left] {{\small $\id_c$}};
%
%
\node[poin] (Kout) [label=above:{\small $c^{\prime\prime}$}] at (3.5,1.25) {};
\node[poin] (Kin) [label=below:{\small $c$}] at (3.5,-0.25) {};
\node[poin] (Kmid) at (3.5,0.5) {};
\draw[thick] (Kin) -- (Kmid) node[midway,left] {{\small $g$}};
\draw[thick] (Kmid) -- (Kout) node[midway,left] {{\small $f$}};
\node[poin] (EqK) at (4,0.5) {$=$};
\node[poin] (KKout) [label=above:{\small $c^{\prime\prime}$}] at (5,1) {};
\node[poin] (KKin) [label=below:{\small $c$}] at (5,0) {};
\draw[thick] (KKin) -- (KKout) node[midway,left] {{\small $f\,g$}};
\end{tikzpicture}
\end{tabular}}
\end{flalign}
for all objects $c\in\CC$ and all pairs of composable $\CC$-morphisms
$(f,g)$. The first relation means that the operadic units are identified with the identity
morphisms of the category $\CC$ and the second relation means that the operadic
composition of $\CC$-morphisms agrees with their categorical composition in $\CC$.

\item {\em Algebra relations:}
\begin{flalign}
\text{\begin{tabular}{c}
\begin{tikzpicture}[cir/.style={circle,draw=black,inner sep=0pt,minimum size=2mm},
        poin/.style={circle, inner sep=0pt,minimum size=0mm},scale=0.8, every node/.style={scale=0.8}]
%
%
\node[poin] (Mout) [label=above:{\small $c$}] at (0,1) {};
\node[poin] (Min1) [label=below:{\small $\emptyset$}] at (-0.3,0) {};
\node[poin] (Min2) [label=below:{\small $c$}] at (0.3,0) {};
\node[poin] (V)  [label=right:{\small $\mu_c$}] at (0,0.5) {};
\draw[thick] (Min1) -- (V) node[pos=0.7,left] {{\small $1_c$}};
\draw[thick] (Min2) -- (V);
\draw[thick] (V) -- (Mout);
\draw[thick,fill=white] (Min1) circle (0.6mm);
\node[poin] (EqM) at (0.75,0.5) {$=$};
\node[poin] (MMout) [label=above:{\small $c$}] at (1.25,1) {};
\node[poin] (MMin) [label=below:{\small $c$}] at (1.25,0) {};
\draw[thick] (MMin) -- (MMout) node[midway,left] {{\small $\oone$}};
\node[poin] (EqMM) at (1.5,0.5) {$=$};
\node[poin] (MMMout) [label=above:{\small $c$}] at (2.25,1) {};
\node[poin] (MMMin1) [label=below:{\small $c$}] at (1.95,0) {};
\node[poin] (MMMin2) [label=below:{\small $\emptyset$}] at (2.55,0) {};
\node[poin] (VVV)  [label=left:{\small $\mu_c$}] at (2.25,0.5) {};
\draw[thick] (MMMin1) -- (VVV);
\draw[thick] (MMMin2) -- (VVV) node[pos=0.7,right] {{\small $1_c$}};
\draw[thick] (VVV) -- (MMMout);
\draw[thick,fill=white] (MMMin2) circle (0.6mm);
%
%
\node[poin] (Nout) [label=above:{\small $c$}] at (4.5,1.25) {};
\node[poin] (Nin1) [label=left:{\small $\mu_c$}] at (4.2,0.25) {};
\node[poin] (Nin2) at (7.3,0) {};
\node[poin] (NV)  [label=right:{\small $\mu_c$}] at (4.5,0.75) {};
\node[poin] (Ninin1) [label=below:{\small $c$}]  at (3.9,-0.25) {};
\node[poin] (Ninin2) [label=below:{\small $c$}]  at (4.5,-0.25) {};
\node[poin] (Ninin3) [label=below:{\small $c$}]  at (5.1,-0.25) {};
\draw[thick] (Nin1) -- (NV);
\draw[thick] (Ninin3) -- (NV);
\draw[thick] (Ninin1) -- (Nin1);
\draw[thick] (Ninin2) -- (Nin1);
\draw[thick] (NV) -- (Nout);
\node[poin] (EqN) at (5.5,0.5) {$=$};
\node[poin] (NNout) [label=above:{\small $c$}] at (6.5,1.25) {};
\node[poin] (NNin2) [label=right:{\small $\mu_c$}] at (6.8,0.25) {};
\node[poin] (NNV)  [label=left:{\small $\mu_c$}] at (6.5,0.75) {};
\node[poin] (NNinin1) [label=below:{\small $c$}]  at (5.9,-0.25) {};
\node[poin] (NNinin2) [label=below:{\small $c$}]  at (6.5,-0.25) {};
\node[poin] (NNinin3) [label=below:{\small $c$}]  at (7.1,-0.25) {};
\draw[thick] (NNinin1) -- (NNV);
\draw[thick] (NNin2) -- (NNV);
\draw[thick] (NNV) -- (NNout);
\draw[thick] (NNinin2) -- (NNin2);
\draw[thick] (NNinin3) -- (NNin2);
\end{tikzpicture}
\end{tabular}}
\end{flalign} 
for all objects $c\in\CC$. These relations express unitality and associativity
of the multiplication of observables in the space-time $c$.

\item {\em Compatibility relations:}
\begin{flalign}
\text{\begin{tabular}{c}
\begin{tikzpicture}[cir/.style={circle,draw=black,inner sep=0pt,minimum size=2mm},
        poin/.style={circle, inner sep=0pt,minimum size=0mm},scale=0.8, every node/.style={scale=0.8}]
%
%
\node[poin] (Kout) [label=above:{\small $c^{\prime}$}] at (0,1.25) {};
\node[poin] (Kin) [label=below:{\small $\emptyset$}] at (0,-0.25) {};
\node[poin] (Kmid) at (0,0.5) {};
\draw[thick] (Kin) -- (Kmid) node[midway,left] {{\small $1_c$}};
\draw[thick] (Kmid) -- (Kout) node[midway,left] {{\small $f$}};
\draw[thick,fill=white] (Kin) circle (0.6mm);
\node[poin] (EqK) at (0.5,0.5) {$=$};
\node[poin] (KKout) [label=above:{\small $c^{\prime}$}] at (1.5,1) {};
\node[poin] (KKin) [label=below:{\small $\emptyset$}] at (1.5,0) {};
\draw[thick] (KKin) -- (KKout) node[midway,left] {{\small $1_{c^\prime}$}};
\draw[thick,fill=white] (KKin) circle (0.6mm);
%
%
\node[poin] (MMout) [label=above:{\small $c^\prime$}] at (4,1.25) {};
\node[poin] (MMin1) [label=below:{\small $c$}] at (3.5,-0.25) {};
\node[poin] (MMin2) [label=below:{\small $c$}] at (4.5,-0.25) {};
\node[poin] (VV)  [label=right:{\small $\mu_{c}$}] at (4,0.5) {};
\draw[thick] (MMin1) -- (VV);
\draw[thick] (MMin2) -- (VV);
\draw[thick] (VV) -- (MMout)  node[midway,left] {{\small $f$}};
\node[poin] (EqM) at (5,0.5) {$=$};
\node[poin] (Mout) [label=above:{\small $c^\prime$}] at (6,1.25) {};
\node[poin] (Min1) [label=below:{\small $c$}] at (5.5,-0.25) {};
\node[poin] (Min2) [label=below:{\small $c$}] at (6.5,-0.25) {};
\node[poin] (V)  [label=right:{\small $\mu_{c^\prime}$}] at (6,0.5) {};
\draw[thick] (Min1) -- (V) node[midway,left] {{\small $f$}};
\draw[thick] (Min2) -- (V) node[midway,right] {{\small $f$}};
\draw[thick] (V) -- (Mout);
\end{tikzpicture}
\end{tabular}}
\end{flalign}
for all $\CC$-morphisms $f\in \mathrm{Mor}\CC$. These relations express
compatibility between the observable algebra structure on each space-time
and pushforwards of observables along space-time embeddings.

\item {\em $\perp$-commutativity relations:}
\begin{flalign}
\text{\begin{tabular}{c}
\begin{tikzpicture}[cir/.style={circle,draw=black,inner sep=0pt,minimum size=2mm},
        poin/.style={circle, inner sep=0pt,minimum size=0mm},scale=0.8, every node/.style={scale=0.8}]
%
%
%
\node[poin] (Mout) [label=above:{\small $c$}] at (0,1.25) {};
\node[poin] (Min1) [label=below:{\small $c_1$}] at (-0.5,-0.25) {};
\node[poin] (Min2) [label=below:{\small $c_2$}] at (0.5,-0.25) {};
\node[poin] (V)  [label=right:{\small $\mu_c$}] at (0,0.5) {};
\draw[thick] (Min1) -- (V) node[midway,left] {{\small $f_1$}};
\draw[thick] (Min2) -- (V) node[midway,right] {{\small $f_2$}};
\draw[thick] (V) -- (Mout);
\node[poin] (Eq) at (1.25,0.5) {$=$};
\node[poin] (MMout) [label=above:{\small $c$}] at (2.5,1.5) {};
\node[poin] (MMin1)  at (2,0) {};
\node[poin] (MMin2)  at (3,0) {};
\node[poin] (MMin11) [label=below:{\small $c_1$}] at (2,-0.5) {};
\node[poin] (MMin21) [label=below:{\small $c_2$}] at (3,-0.5) {};
\node[poin] (VV)  [label=right:{\small $\mu_c$}] at (2.5,0.75) {};
\draw[thick] (MMin1) -- (VV) node[midway,left] {{\small $f_2$}};
\draw[thick] (MMin2) -- (VV) node[midway,right] {{\small $f_1$}};
\draw[thick] (VV) -- (MMout);
\draw[thick] (MMin11) -- (MMin2);
\draw[thick] (MMin21) -- (MMin1);
\end{tikzpicture}
\end{tabular}}
\end{flalign}
for all orthogonal pairs $f_1\perp f_2$ of $\CC$-morphisms.
These are the key relations that encode the $\perp$-commu\-tativity
property of Definition \ref{def:QFTcats} as a {\em{\underline{structure}}} 
into the colored operad $\O_{\ovr{\CC}}$.
\end{enumerate}

As an alternative to this graphical description
of the AQFT operad $\O_{\ovr{\CC}}$ by generators and relations, there 
is the following more algebraic description. See \cite{Benini:2017fnn}
for a proof of the following theorem.
\begin{theo}\label{theo:AQFToperad}
The AQFT operad $\O_{\ovr{\CC}}\in\Op(\Set)$ described
above is (isomorphic to) the colored operad specified by
the following data:
\begin{enumerate}[i)]
\item the set of colors is the set of objects of  $\CC$;
\item the set of operations from $\und{c} = (c_1,\dots,c_n)$ to
$t$ is the quotient set
\begin{flalign}
\O_{\ovr{\CC}}\big(\substack{t \\ \und{c}}\big) = \Big(\Sigma_n \times \prod_{i=1}^n\CC(c_i,t)\Big)\Big/ \sim_{\perp}\quad,
\end{flalign}
where $\Sigma_n$ is the permutation group on $n$ letters, $\CC(c_i,t)$ are
the $\Hom$-sets of $\CC$, and the equivalence relation is defined as follows:
$(\sigma,\und{f}) \sim_\perp (\sigma^\prime, \und{f}^\prime)$ if and only if
$\und{f} = \und{f}^\prime$ and the right permutation
$\sigma\sigma^{\prime\,-1} : \und{f}\sigma^{-1}\to \und{f}\sigma^{\prime\,-1}$
is generated by transpositions of adjacent orthogonal pairs;

\item the compositions $\gamma : \O_{\ovr{\CC}}\big(\substack{t \\ \und{c}} \big)\times\prod_{i=1}^n\O_{\ovr{\CC}}\big(\substack{c_i \\ \und{b}_i} \big) \to$\linebreak $\O_{\ovr{\CC}}\big(\substack{t \\ (\und{b}_1,\dots,\und{b}_n)} \big)$ are
\begin{multline}
\gamma\big([\sigma,\und{f}],\big([\sigma_1,\und{g}_1],\dots,[\sigma_n,\und{g}_n]\big)\big)\\
~=~\big[\sigma(\sigma_1,\dots,\sigma_n), \und{f}(\und{g}_1,\dots,\und{g}_n)\big]\quad,
\end{multline}
where $\sigma(\sigma_1,\dots,\sigma_n) = \sigma\langle k_{\sigma^{-1}(1)},\dots, k_{\sigma^{-1}(n)}\rangle 
~(\sigma_1\oplus\cdots\oplus\sigma_n)$ is the product of the block permutation induced by $\sigma$
and the sum permutation induced by the $\sigma_i$, where $k_i$ is the length of the tuple $\und{b}_i$,
and $\und{f}(\und{g}_1,\dots,\und{g}_n) = \big(f_1\,g_{11},\dots, f_1\,g_{1k_1},\dots, f_n\,g_{n1},\dots,f_n\,g_{nk_n}\big)$
is given by composition in the category $\CC$;

\item the units are $[e,\id_t]\in \O_{\ovr{\CC}}\big(\substack{t \\ t}\big)$, where 
$e\in\Sigma_1$ is the identity permutation and $\id_t:t\to t$ the identity morphism in $\CC$;

\item the permutation actions are $\O_{\ovr{\CC}}(\sigma^\prime): [\sigma,\und{f}]\mapsto$\linebreak $[\sigma\sigma^\prime,\und{f}\sigma^\prime]$.
\end{enumerate}
\end{theo}

The relevance of the AQFT operad $\O_{\ovr{\CC}}\in\Op(\Set)$ is that its category
of algebras is precisely the category $\QFT(\ovr{\CC})$ of AQFTs defined in Definition
\ref{def:QFTcats}. But what are algebras over operads?
Loosely speaking, an $\MM$-valued algebra over a colored operad $\O\in\Op(\Set)$
is something like a `representation' of the operations described by $\O$ as $\MM$-morphisms
between a colored family of objects $A_c \in\MM$, for all colors $c$. The following 
picture visualizes this basic idea:
\begin{flalign}
\xymatrix{
\parbox{1.6cm}{\begin{tikzpicture}
\draw[fill=gray!10] (0,0.2) -- (0.5,-0.2) -- (-0.5,-0.2) -- (0,0.2);
\draw (0,0) node {{\footnotesize $o$}};
\draw[thick] (0,0.2) -- (0,0.5);
\draw (0,0.7) node {{\footnotesize $c^\prime$}};
\draw[thick] (-0.4,-0.2) -- (-0.4,-0.5);
\draw (-0.4,-0.65) node {{\footnotesize $c_1$}};
\draw[thick] (0.4,-0.2) -- (0.4,-0.5);
\draw (0.4,-0.65) node {{\footnotesize $c_n$}};
\draw (0,-0.5) node {{\footnotesize $\cdots$}};
\end{tikzpicture}}
\ar@{~>}[r]^-{\text{represent}}&
~\quad ~\Big(\bigotimes\limits_{i=1}^n A_{c_i} \ar[r]^-{A(o)} ~& ~ A_{c^\prime}\Big)
}
\end{flalign}
Observe how the input and output colors match on both sides of this picture.
There are of course certain compatibility conditions to be fulfilled,
namely these $\MM$-morphisms must be compatible with operadic compositions,
operadic units and the permutation actions. We refer to e.g.\ \cite{yau2016colored} and \cite{Benini:2017fnn}
for a detailed definition of algebras over colored operads.
The main theorem justifying the relevance of the colored operads $\O_{\ovr{\CC}}\in\Op(\Set)$
is as follows. See \cite{Benini:2017fnn} for a proof.
\begin{theo}\label{theo:AQFTalgebraoperad}
For every orthogonal category $\ovr{\CC}$, the category
$\Alg_{\O_{\ovr{\CC}}}(\MM)$ of algebras over the AQFT operad $\O_{\ovr{\CC}}\in\Op(\Set)$
is (isomorphic to) the category $\QFT(\ovr{\CC})$ of $\MM$-valued AQFTs on $\ovr{\CC}$ from
Definition \ref{def:QFTcats}.
\end{theo}

\begin{rem}
It is instructive to have a closer look at the action of the operations
$[\sigma,\und{f}]\in \O_{\ovr{\CC}}\big(\substack{t \\ \und{c} }\big)$
of the AQFT operad on a $\perp$-commutative
functor $\AAA :\CC\to \Alg_{\As}(\MM)$. One observes that the corresponding
$\MM$-morphism
\begin{subequations}
\begin{flalign}
\AAA\big([\sigma,\und{f}]\big) \,:\, \bigotimes\limits_{i=1}^n \AAA(c_i)~\longrightarrow ~\AAA(t)
\end{flalign}
is concretely given by
\begin{multline}
\AAA\big([\sigma,\und{f}]\big) \big(a_1\otimes\cdots \otimes a_n\big) \\
=\AAA(f_{\sigma^{-1}(1)})\big(a_{\sigma^{-1}(1)}\big)\cdots \AAA(f_{\sigma^{-1}(n)})\big(a_{\sigma^{-1}(n)}\big)\quad.
\end{multline}
\end{subequations}
In words, this means that one first pushes forward each observable $a_i\in\AAA(c_i)$
along the space-time embedding $f_i :c_i\to t$ and then multiplies the
resulting observables in $\AAA(t)$ according to the order prescribed by the permutation $\sigma$.
Note that this is well-defined on the equivalence classes of operations (cf.\ Theorem \ref{theo:AQFToperad})
because $\AAA$ is a $\perp$-commutative functor, i.e.\ one is allowed to interchange the order
of multiplication for observables obtained by pushforward along orthogonal pairs of $\CC$-morphisms.
\end{rem}


\subsection{\label{subsec:universal}Universal constructions}
The result of Theorem \ref{theo:AQFTalgebraoperad}
that AQFTs are algebras over a colored operad $\O_{\ovr{\CC}}\in\Op(\Set)$
is very useful for studying universal constructions in AQFT.
For our first observation, let us recall
that the target category $\MM$ is by assumption
bicomplete.
\begin{propo}
For every orthogonal category $\ovr{\CC}$,
the category $\QFT(\ovr{\CC})$ of $\MM$-valued AQFTs on $\ovr{\CC}$
admits all small limits and colimits. 
\end{propo}
\begin{proof}
By Theorem \ref{theo:AQFTalgebraoperad}, we have that 
$\QFT(\ovr{\CC})\cong \Alg_{\O_{\ovr{\CC}}}(\MM)$ is the category
of algebras over a colored operad. The latter category is known to be bicomplete
whenever the target category $\MM$ is bicomplete.
See e.g.\ \cite[Proposition 1.3.6]{fresse2017homotopy} for a direct proof in the case of 
uncolored operads, which generalizes easily to colored operad.
\end{proof}

Using this result, it is possible to build new AQFTs on $\ovr{\CC}$
by forming limits or colimits of diagrams of AQFTs on the same orthogonal category 
$\ovr{\CC}$. We expect that such constructions might be relevant for formalizing
the algebraic adiabatic limit in perturbative AQFTs (see e.g.\ \cite{Brunetti:2009qc}), however we did not look
into this in any detail. 

Let us now consider a more interesting class of universal constructions
that relate AQFTs on different orthogonal categories. Let us recall
that such constructions are relevant e.g.\ for discussing descent,
which involves AQFTs on both the space-time category $\ovr{\Loc}$
and the category $\ovr{\Locc}$ of space-times whose 
underlying manifold is diffeomorphic to $\bbR^m$. 
The key observation that allows us to develop such constructions
is that the assignment $\ovr{\CC}\mapsto \O_{\ovr{\CC}}$ 
of the AQFT operad to an orthogonal category is functorial $\O_{(-)}:\OCat\to \Op(\Set)$.
This means that given any orthogonal functor $F : \ovr{\CC}\to \ovr{\DD}$, 
we obtain a colored operad morphism $\O_F : \O_{\ovr{\CC}}\to \O_{\ovr{\DD}}$
and hence a pullback functor $\O_F^{\ast} : \Alg_{\O_{\ovr{\DD}}}(\MM) \to \Alg_{\O_{\ovr{\CC}}}(\MM)$
between the corresponding categories of algebras.
Under the natural identification $\QFT(-)\cong \Alg_{\O_{(-)}}(\MM)$ given in 
Theorem \ref{theo:AQFTalgebraoperad}, the pullback functor $\O_F^{\ast} $
gets identified with the restriction of the pullback functor
$F^\ast := (-)\circ F : \Alg_{\As}(\MM)^\DD \to \Alg_{\As}(\MM)^\CC$ 
to the full subcategories of AQFTs from Definition \ref{def:QFTcats}. This functor
always admits a left adjoint.
\begin{theo}\label{theo:algadjunction}
For every orthogonal functor $F : \ovr{\CC}\to \ovr{\DD}$, there exists
an adjunction
\begin{flalign}
\xymatrix{
F_! \,:\, \QFT(\ovr{\CC}) ~\ar@<0.5ex>[r] & \ar@<0.5ex>[l]~ \QFT(\ovr{\DD}) \,:\, F^\ast
}\quad,
\end{flalign}
where the right adjoint $F^\ast$ is the restriction of the pullback functor on functor categories
to the subcategories of AQFTs.
\end{theo}
\begin{proof}
The operadic pullback functor $\O_F^{\ast} : \Alg_{\O_{\ovr{\DD}}}(\MM) \to \Alg_{\O_{\ovr{\CC}}}(\MM)$
admits a left adjoint given by operadic left Kan extension, see e.g.\ \cite{Benini:2017fnn} for a brief review.
The natural identification of Theorem \ref{theo:AQFTalgebraoperad} then proves our claim. 
\end{proof}

Using the adjunctions from Theorem \ref{theo:algadjunction},
one can develop and study interesting universal constructions that
relate AQFTs on different orthogonal categories. Before 
looking at concrete examples inspired by physics, let us first 
note the following structural result for the full subcategory
$\QFT(\ovr{\CC}) \subseteq \Alg_{\As}(\MM)^\CC$ of AQFTs
from Definition \ref{def:QFTcats}. Given any orthogonal category
$\ovr{\CC} = (\CC,\perp)$, we form the orthogonal category
$(\CC,\emptyset)$ with the trivial (empty) orthogonality relation
and observe that the identity functor defines an orthogonal
functor $p_{\ovr{\CC}} :=\id_\CC : (\CC,\emptyset) \to \ovr{\CC}$. Note that
$\QFT(\CC,\emptyset) = \Alg_{\As}(\MM)^\CC$ is the functor category.
The following result is immediate.
\begin{propo}\label{propo:QFTreflective}
For every orthogonal category $\ovr{\CC}$, 
the adjunction
\begin{flalign}
\xymatrix{
{p_{\ovr{\CC}}}_! \,:\,  \Alg_{\As}(\MM)^\CC ~\ar@<0.5ex>[r] & \ar@<0.5ex>[l]~ \QFT(\ovr{\CC}) \,:\, {p_{\ovr{\CC}}}^\ast
}
\end{flalign}
corresponding to the canonical orthogonal functor $p_{\ovr{\CC}}=\id_\CC : (\CC,\emptyset) \to \ovr{\CC}$
exhibits $\QFT(\ovr{\CC})$ as a full reflective subcategory of the functor category
$ \Alg_{\As}(\MM)^\CC$.
\end{propo}

A corollary of this result is that the left adjoint functor $F_!$ 
from Theorem \ref{theo:algadjunction} can be related to the ordinary left Kan extension
$\mathrm{Lan}_F : \Alg_{\As}(\MM)^\CC\to \Alg_{\As}(\MM)^\DD$ of $\Alg_{\As}(\MM)$-valued functors.
\begin{cor}
For every orthogonal functor $F : \ovr{\CC}\to \ovr{\DD}$, there exists
a natural isomorphism of functors
\begin{flalign}
F_!\,\cong\, {p_{\ovr{\DD}}}_!\circ \mathrm{Lan}_F\circ {p_{\ovr{\CC}}}^\ast ~:~ \QFT(\ovr{\CC})\longrightarrow\QFT(\ovr{\DD})\quad,
\end{flalign}
where $F_! $ is the left adjoint from Theorem \ref{theo:algadjunction},
$\mathrm{Lan}_F $ the ordinary categorical left Kan extension and $p_!\dashv p^\ast$ the adjunction from
Proposition \ref{propo:QFTreflective}.
\end{cor}

Our first concrete example for an adjunction as in Theorem \ref{theo:algadjunction}
is motivated physically by our goal to introduce a descent condition for AQFTs.
Let $\ovr{\DD}$ be any orthogonal category and $\ovr{\CC}\subseteq\ovr{\DD}$ 
a {\em full orthogonal subcategory}, i.e.\ $\CC\subseteq\DD$ is a full subcategory
such that $f_1\perp_\CC f_2$ if and only if $f_1\perp_\DD f_2$. 
For instance, $\ovr{\DD}$ may be the space-time category $\ovr{\Loc}$ from Example
\ref{ex:LocLocc} and $\ovr{\CC}$ the category $\ovr{\Locc}$ of space-times
whose underlying manifold is diffeomorphic to $\bbR^m$. In general,
one should interpret $\ovr{\DD}$ as a category of `all space-times'
and $\ovr{\CC}\subseteq \ovr{\DD}$ as a full subcategory of `nice space-times'.
Embedding the full subcategory defines an orthogonal functor that we shall denote
by $j: \ovr{\CC}\to \ovr{\DD}$. As a consequence of Theorem
\ref{theo:algadjunction}, we obtain an adjunction
\begin{flalign}\label{eqn:jadjunction}
\xymatrix{
j_! \,:\, \QFT(\ovr{\CC}) ~\ar@<0.5ex>[r] & \ar@<0.5ex>[l]~ \QFT(\ovr{\DD}) \,:\, j^\ast
}
\end{flalign}
between AQFTs on `nice space-times' and AQFTs on `all space-times'.
The right adjoint functor $j^\ast$ should be interpreted as a restriction
functor that restricts an AQFT $\AAA\in\QFT(\ovr{\DD})$ that is defined
on all of $\ovr{\DD}$ to an AQFT $j^\ast \AAA \in\QFT(\ovr{\CC})$ on
the full orthogonal subcategory $\ovr{\CC}$. More interestingly, the
left adjoint $j_! $ is a universal extension functor that extends
an AQFT $\BBB\in\QFT(\ovr{\CC})$ that is defined only on `nice space-times'
in $\ovr{\CC}\subseteq\ovr{\DD}$ to all of $\ovr{\DD}$.
In contrast to Fredenhagen's universal algebra construction \cite{Fredenhagen:1989pw,Fredenhagen:1993tx,Fredenhagen:1992yz,lang2014universal},
which is given by left Kan extension $\mathrm{Lan}_j : \Alg_{\As}(\MM)^\CC\to\Alg_{\As}(\MM)^\DD$
of the underlying functors, our left adjoint $j_!$ always defines an AQFT $j_!\BBB\in\QFT(\ovr{\DD})$
on $\ovr{\DD}$ and not only a functor $\mathrm{Lan}_j \BBB : \DD\to \Alg_{\As}(\MM)$
that might violate the $\perp$-commutativity axiom from Definition \ref{def:QFTcats}.
(It was shown in \cite{Benini:2017fnn} that $\mathrm{Lan}_j \BBB $ violates
$\perp$-commutativity on non-connected space-times.) The following
result states an important technical property of the adjunction \eqref{eqn:jadjunction}.
\begin{propo}\label{prop:jcorefl}
For every full orthogonal subcategory embedding
$j : \ovr{\CC}\to\ovr{\DD}$, the adjunction \eqref{eqn:jadjunction}
exhibits $\QFT(\ovr{\CC})$ as a full coreflective subcategory of
$\QFT(\ovr{\DD})$.
\end{propo}

\begin{rem}
Let us explain in more detail why this result is crucial for the interpretation
of $j^\ast$ as a restriction functor and $j_!$ as an extension functor.
Given an AQFT $\BBB \in\QFT(\ovr{\CC})$,
we form the extension $j_! \BBB \in\QFT(\ovr{\DD})$ and ask if this
alters the values of the AQFT on the subcategory $\ovr{\CC}\subseteq\ovr{\DD}$,
i.e.\ if the restriction $j^\ast j_!\BBB$ of the extension is isomorphic to the original theory $\BBB$.
Proposition \ref{prop:jcorefl} states that this is the case and that the unit
$\eta_\BBB : \BBB\to j^\ast j_!\BBB$ of the adjunction  \eqref{eqn:jadjunction}
provides such an isomorphism. In other words, the extension functor
$j_!$ does not alter the values of AQFTs on the full subcategory $\ovr{\CC}\subseteq\ovr{\DD}$.
\end{rem}

We now can formalize our sketchy ideas from Section \ref{sec:background} 
about a descent condition in AQFT.
\begin{defi}\label{def:jlocal}
An AQFT $\AAA \in \QFT(\ovr{\DD})$ is called {\em $j$-local}
if the corresponding component $\epsilon_{\AAA} : j_! j^\ast\AAA \to \AAA$ 
of the counit of the adjunction \eqref{eqn:jadjunction} is an isomorphism.
We denote the full subcategory of $j$-local AQFTs by\linebreak
$\QFT(\ovr{\DD})^{j-\mathrm{loc}} \subseteq \QFT(\ovr{\DD})$.
\end{defi}

\begin{cor}
For every full orthogonal subcategory embedding
$j : \ovr{\CC}\to\ovr{\DD}$, the adjunction \eqref{eqn:jadjunction}
restricts to an adjoint equivalence
$j_! : \QFT(\ovr{\CC}) \rightleftarrows \QFT(\ovr{\DD})^{j-\mathrm{loc}} : j^\ast$.
\end{cor}

\begin{rem}
The physical interpretation is that
$j$-local AQFTs $\AAA\in \QFT(\ovr{\DD})^{j-\mathrm{loc}}$
are those AQFTs on $\ovr{\DD}$ that are completely determined
by their restriction to the subcategory $\ovr{\CC}\subseteq \ovr{\DD}$. 
In the case of $\ovr{\Locc}\subseteq\ovr{\Loc}$, this means that
the value $\AAA(M)$ of a $j$-local AQFT on {\em any} space-time
$M\in\Loc$ is completely determined by the values of $\AAA$
on the subcategory $\Locc$ of space-times with underlying manifold
diffeomorphic to $\bbR^m$. Hence, $j$-locality is a type of descent 
condition for AQFTs.
\end{rem}

\begin{ex}\label{ex:jlocal}
From the results in \cite{lang2014universal} and \cite{Benini:2017fnn}
one can conclude that the free Klein-Gordon AQFT
is $j$-local in the above sense for $j: \ovr{\Locc}\to\ovr{\Loc}$.
\end{ex}

Our second concrete example for an adjunction as in Theorem \ref{theo:algadjunction}
is motivated by the time-slice axiom of AQFTs. Let $\ovr{\CC}$ be any orthogonal category
and $W\subseteq \mathrm{Mor} \CC$ any subset of the set of morphisms.
The corresponding {\em orthogonal localization} functor $L : \ovr{\CC}\to \ovr{\CC[W^{-1}]}$
(see the text before Proposition \ref{prop:timeslicelocalization}) defines
an adjunction
\begin{flalign}\label{eqn:Ladjunction}
\xymatrix{
L_! \,:\, \QFT(\ovr{\CC}) ~\ar@<0.5ex>[r] & \ar@<0.5ex>[l]~ \QFT(\ovr{\CC[W^{-1}]}) \,:\, L^\ast
}
\end{flalign}
between AQFTs on $\ovr{\CC}$ and AQFTs on $\ovr{\CC[W^{-1}]}$.
The following result is a direct consequence of Proposition \ref{prop:timeslicelocalization}.
\begin{propo}\label{propo:Lreflective}
For every orthogonal localization $L : \ovr{\CC}\to \ovr{\CC[W^{-1}]}$,
the right adjoint functor $L^\ast$ in \eqref{eqn:Ladjunction} is fully faithful
and its essential image is the full subcategory $\QFT(\ovr{\CC})^{W-\mathrm{const}}\subseteq
\QFT(\ovr{\CC})$ of $W$-constant AQFTs. Hence, the adjunction \eqref{eqn:Ladjunction}
exhibits $ \QFT(\ovr{\CC[W^{-1}]})$ as a full reflective subcategory of $\QFT(\ovr{\CC})$
and it restricts to an adjoint equivalence $L_! : \QFT(\ovr{\CC})^{W-\mathrm{const}} \rightleftarrows
\QFT(\ovr{\CC[W^{-1}]}) : L^\ast$.
\end{propo}

An immediate corollary of this result is that there exist
equivalent characterizations of $W$-constant AQFTs.
\begin{cor}\label{cor:Wconstancy}
Let $\AAA\in\QFT(\ovr{\CC})$. Then the following are equivalent:
\begin{enumerate}[i)]
\item $\AAA$ is $W$-constant, i.e.\ for all $f\in W$ the $\Alg_{\As}(\MM)$-morphism
$\AAA(f) : \AAA(c)\to \AAA(c^\prime)$ is an isomorphism.
\item The component $\eta_{\AAA} : \AAA \to L^\ast L_!\AAA$ of the unit
of the adjunction \eqref{eqn:Ladjunction} is an isomorphism.
\end{enumerate}
\end{cor}

\begin{rem}
We would like to stress that our adjunctions
from this section are not only theoretically interesting,
but they already found concrete applications to physical problems.
We refer to \cite{Benini:2017dfw} for a 
study of AQFTs on space-times with time-like boundaries 
from this perspective.
\end{rem}


\section{\label{sec:gaugetheory}Higher structures in gauge theory}
The aim of this section is to explain in rather non-technical terms the structural 
differences between `ordinary' field theories, such as Klein-Gordon theory, 
and gauge theories, such as Yang-Mills theory. The latter are instances 
of {\em higher structures} and therefore require refined concepts of 
category theory for their formalization, e.g.\ $\infty$-category theory \cite{Lurie:0608040,Lurie:book}
or model category theory \cite{hovey2007model,dwyer1995homotopy}. The incorporation
of such higher structures into AQFT will be discussed in Section \ref{sec:homotopy}.


\subsection{\label{subsec:groupoids}Groupoids of gauge fields}
The main difference between `ordinary' field theories
and gauge theories is of course the presence of 
gauge symmetries. Even though this observation
sounds like a tautology, it has profound consequences
on how one should think of the spaces of fields in these 
situations. For `ordinary' theories, the collection
of all fields $\mathfrak{F}$ has the structure of a set, i.e.\ 
given two fields $\Phi,\Phi^\prime\in\mathfrak{F}$, we can
ask if these fields are the same or not by testing whether 
$\Phi = \Phi^\prime$ holds true. In a gauge theory
this becomes more complicated because in addition
to gauge fields $A$ one also has gauge transformations
$A \stackrel{g}{\longrightarrow}A^\prime$ between gauge fields.
Hence, the collection of all gauge fields $\mathfrak{G}$ has the structure of
a groupoid and not that of a set! The following picture visualizes the basic idea:
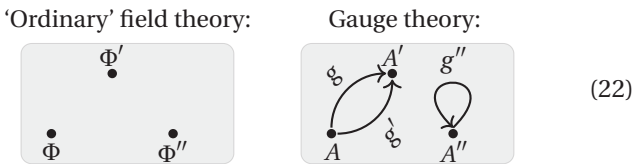
\begin{flalign}\label{eqn:gaugetheorypic}
\text{\begin{tabular}{cc}
{\small `Ordinary' field theory:} & {\small Gauge theory:}\\
~~\begin{tikzpicture}[scale=0.80]
\draw [fill=gray,opacity=0.15, rounded corners] (-2.5,-1) rectangle (1,1); 
\draw [fill=black,thick] (-2,-0.5) circle (0.06);
\draw [fill=black,thick] (0,-0.5) circle (0.06);
\draw [fill=black,thick] (-1,0.5) circle (0.06);
\draw (-2,-0.8) node {{\small {$\Phi$}}};
\draw (-1,0.8) node {{\small {$\Phi^\prime$}}};
\draw (0,-0.8) node {{\small {$\Phi^{\prime\prime}$}}};
\end{tikzpicture}~~
&~~
\begin{tikzpicture}[scale=0.80]
\draw [fill=gray,opacity=0.15, rounded corners] (-2.5,-1) rectangle (1,1); 
\draw [fill=black,thick] (-2,-0.5) circle (0.06);
\draw [fill=black,thick] (0,-0.5) circle (0.06);
\draw [fill=black,thick] (-1,0.5) circle (0.06);
\draw (-2,-0.8) node {{\small {$A$}}};
\draw (-1,0.8) node {{\small {$A^\prime$}}};
\draw (0,-0.8) node {{\small {$A^{\prime\prime}$}}};
\draw [thick,->] (-2,-0.4) to[out=90, in=180,edge node={node [sloped,above] {\small {$g$}}}] (-1.1,0.5);
\draw [thick,->] (-1.9,-0.5) to[out=0, in=270 ,edge node={node [sloped,below] {\small {$g^\prime$}}}] (-1,0.4);
\draw [thick,->] (0,-0.4) to[loop,distance=40,edge node={node [sloped,above] {\small {$g^{\prime\prime}$}}}] (0,-0.4);
\end{tikzpicture}~~
\end{tabular}}
\end{flalign}
This groupoid structure drastically changes the way one 
should think of two gauge fields as being the same. In contrast to sets, being 
the same in a groupoid is not anymore a property, but rather a structure
in the sense that one needs a gauge transformation $A\stackrel{g}{\longrightarrow} A^\prime$
in order to witness that $A$ and $A^\prime$ are the same.
As visualized in \eqref{eqn:gaugetheorypic}, there might exist different 
witnesses for two gauge fields being the same, and in particular there are
generically non-trivial `loops' in the groupoid of gauge fields.
These loops should be understood as higher order structures in the groupoid
of gauge fields that cannot be seen at the level of the naive `gauge orbit space'. 
Recall that the naive `gauge orbit space' is obtained by forming gauge equivalence classes
of gauge fields, i.e.\ it is the zeroth homotopy group
$\pi_0\mathfrak{G}$ of the groupoid of gauge fields $\mathfrak{G}$. This construction
however neglects information on the loops in $\mathfrak{G}$, which is
contained in the first homotopy groups $\pi_1(\mathfrak{G},A)$, for $A\in\mathfrak{G}$. 
Hence, the groupoid
of gauge fields $\mathfrak{G}$ includes more refined information on the gauge theory
than the naive `gauge orbit space'. We shall explain later why this additional
information is crucial.
\begin{rem}\label{rem:setgrpdinftygrpd}
The same way of reasoning of course also applies to gauge transformations themselves.
In particular, if there are gauge transformations of gauge transformations,
then the collection of gauge fields is described by a $2$-groupoid.
If there are gauge transformations of gauge transformations of gauge transformations,
then one gets a $3$-groupoid, and so on. Hence, the natural framework
in which to study gauge theories and higher gauge theories 
is that of $\infty$-groupoids. Because of the chain of inclusions
\begin{flalign}
\Set \hookrightarrow \mathbf{Grpd} \hookrightarrow 2\mathbf{Grpd}  \hookrightarrow \cdots \hookrightarrow\infty\mathbf{Grpd}\quad,
\end{flalign}
all ordinary field theories, gauge theories, $2$-gauge theories, \dots, can be regarded as 
particular examples of theories in the sense of $\infty$-groupoids. In what follows we 
will mostly focus on the case of $1$-groupoids in order to simplify our presentation. 
However, everything said below generalizes to $\infty$-groupoids and in particular
our model categorical framework for AQFT in Section \ref{sec:homotopy} applies
to higher gauge theories as well.
\end{rem}

\begin{ex}\label{ex:BGcon}
As a very concrete example, let us consider principal $G$-bundles
with connections on a Cartesian space $U\cong \bbR^m$. Because all principal
$G$-bundles on $U$ are trivializable, the groupoid of gauge fields on $U$ is
\begin{flalign}
BG^\mathrm{con}(U) \,=\,\begin{cases}
\text{Obj:} &  A\in \Omega^1(U,\mathfrak{g})\\
\text{Mor:} & A\stackrel{g}{\longrightarrow} A \triangleleft g := g^{-1}A g + g^{-1}\dd g \\
& \text{with }  g\in C^\infty(U,G)
\end{cases}\quad\quad,
\end{flalign}
where $\mathfrak{g}$ is the Lie algebra of the structure Lie group $G$.
For the Abelian cases $G = U(1)$ or $G=\bbR$, one easily computes the
homotopy groups and obtains $\pi_0 BG^\mathrm{con}(U) \cong \Omega^1(U)\big/ \dd \Omega^0(U)$
and  $\pi_1 (BG^\mathrm{con}(U), A) \cong G$. Hence, the naive 
`gauge orbit space' does not distinguish between the two different structure groups,
but the higher order information contained in $\pi_1$ does. In other words, 
the groupoid perspective on gauge theory is truly more refined
than the naive `gauge orbit space' perspective.
\end{ex}

Working with groupoids requires some additional
care because the correct notion of two groupoids being the same
is via categorical equivalence rather than isomorphism.
This is because the category of groupoids $\Grpd$ is actually
a $2$-category, whose objects are all\linebreak groupoids $\mathfrak{G}$,
$1$-morphisms are functors $F : \mathfrak{G}\to \mathfrak{G}^\prime$ and
$2$-morphisms are natural isomorphisms $\zeta : F \to F^\prime$ between
functors $F,F^\prime : \mathfrak{G}\to \mathfrak{G}^\prime$.
The equivalences in this $2$-category are the usual categorical equivalences, i.e.\
functors $F : \mathfrak{G}\to \mathfrak{G}^\prime$ that can be `inverted up to $2$-morphisms'
in the sense that there exists a functor $F^\prime : \mathfrak{G}^\prime\to \mathfrak{G}$,
going in the opposite direction, together with natural isomorphisms 
$F^\prime \,F \cong \id_{\mathfrak{G}}$ and $F\,F^\prime \cong \id_{\mathfrak{G}^\prime}$.
Recall that these equivalences can be characterized as fully faithful and essentially surjective 
functors $F : \mathfrak{G}\to\mathfrak{G}^\prime$. As a side-remark, let us briefly mention
that the category of $\infty$-groupoids is not only a $2$-category, but actually 
an $\infty$-category \cite{Lurie:0608040}. Hence, the higher the gauge theory one considers, 
the higher one has to climb up on the categorical ladder.

For our purposes, it will be convenient to adopt a slightly different, but related,  
point of view and regard $\Grpd$ (and also $\infty\Grpd$) as a {\em model category}.
A model category is a bicomplete category
$\CC$ that is endowed with three distinguished classes of morphisms -- called
weak equivalences, fibrations and cofibrations -- that have to satisfy a list of conditions.
See e.g.\ \cite{dwyer1995homotopy,hovey2007model} for details. These axioms are designed in
such a way that the weak equivalences define a consistent notion of two 
objects being the same. In particular, this notion is preserved
under certain (derived) functorial constructions. Let us expand on the latter point 
because it is crucial. Given any functor $F : \CC\to\DD$ between two model 
categories $\CC$ and $\DD$, it is in general
not true that $F$ maps weak equivalences in $\CC$ to weak equivalences in $\DD$. 
This will of course introduce inconsistencies, because
weakly equivalent objects are regarded as being the same 
according to the philosophy of model category theory. 
Considering only those functors $F : \CC\to \DD$ that do preserve
weak equivalences would be too restrictive, because several 
natural constructions, e.g.\ limit and colimit functors, are not of this type.
The way out of this dilemma is to `deform' (in a controlled way) the functor
$F : \CC\to \DD$ to obtain a functor that does preserve weak equivalences. That
is precisely what derived functors do for us! The usual context in which the
theory of derived functors applies is when one has a {\em Quillen adjunction}
between model categories, i.e.\ an adjunction $F : \CC \rightleftarrows \DD : G$ in
which the right adjoint functor $G$ preserves fibrations and acyclic fibrations (i.e.\
morphisms that are both a fibration and a weak equivalence). Choosing
a natural cofibrant replacement $(Q : \CC\to\CC, q : Q\stackrel{\sim}{\to}\id_\CC)$ 
for $\CC$ and a natural fibrant replacement $(R : \DD\to\DD, r : \id_\DD \stackrel{\sim}{\to} R)$ for $\DD$,
one can define the {\em left derived functor}
\begin{subequations}
\begin{flalign}
\bbL F := F\,Q \,:\, \CC\longrightarrow\DD
\end{flalign}
and the {\em right derived functor}
\begin{flalign}
\bbR G := G\,R\,:\, \DD\longrightarrow \CC
\end{flalign}
\end{subequations}
corresponding to the Quillen adjunction $F \dashv G$. It can be shown that 
both derived functors preserve weak equivalences and that different
choices of (co)fibrant replacements define naturally weakly equivalent 
derived functors \cite{dwyer1995homotopy,hovey2007model}.

Let us now look at model categories and derived functors in action in order
to better understand what they do for us and why they are crucial.
We first recall that the category $\Grpd$ of groupoids is a model category
with respect to the following choices (see e.g.\ \cite{Hollander:0110247}): 
A morphism $F : \mathfrak{G} \to\mathfrak{H}$
(i.e.\ functor) between two groupoids is
\begin{enumerate}[i)]
\item a weak equivalence if it is fully faithful and essentially surjective;
\item a fibration if it is an isofibration, i.e.\ for each object $x\in\mathfrak{G}$ 
and $\mathfrak{H}$-morphism $g : F(x)\to y$ there exists a $\mathfrak{G}$-morphism
$f : x\to x^\prime$ such that $F(f) = g$;
\item a cofibration if it is injective on objects.
\end{enumerate}
Note that the weak equivalences in this model structure are precisely
the equivalences one obtains when thinking of $\Grpd$ as a $2$-category.
Given any small category $\DD$, we consider the functor category
$\Grpd^\DD$ of all functors from $\DD$ to $\Grpd$, which we interpret as diagrams
of shape $\DD$ in $\Grpd$. The constant diagram functor $\mathrm{const} : 
\Grpd \to \Grpd^\DD$ admits both a left and a right adjoint functor, 
respectively given by the colimit functor $\colim : \Grpd^\DD\to \Grpd$ and the limit functor
$\lim : \Grpd^\DD\to \Grpd$. Let us focus on the adjunction
$\mathrm{const} \dashv \lim$, whose right adjoint is the limit functor,
and note that this is a Quillen adjunction when one endows $\Grpd^\DD$ with
the injective model category structure, i.e.\ a  $\Grpd^\DD$-morphism
$\zeta : X\to Y$  is a cofibration (respectively, a weak equivalence) if
all components $\zeta_d : X(d)\to Y(d)$, for $d\in\DD$, are cofibrations
(respectively,  weak equivalences) in $\Grpd$. The corresponding 
right derived functor
\begin{flalign}
\holim := \bbR \lim \,:\, \Grpd^\DD\longrightarrow \Grpd 
\end{flalign}
is called the {\em homotopy limit functor}. In contrast to the ordinary 
limit functor $\lim$, the homotopy limit functor $\holim$ has the important
property that it preserves weak equivalences. 

For our gauge-theoretic example below, we shall need a concrete model for homotopy limits
of cosimplicial groupoids, see e.g.\ \cite{Hollander:0110247} for details. 
Let $\DD = \Delta$ be the simplex category and consider
the corresponding functor category $\Grpd^\Delta$. An object
$\mathfrak{G}^\bullet \in \Grpd^\Delta$ is a cosimplicial groupoid,
which one may visualize as follows
\begin{flalign}
\mathfrak{G}^\bullet = \Bigg(
\xymatrix@C=1.5em{
\mathfrak{G}^0 \ar@<0.5ex>[r]^-{d^0}\ar@<-0.5ex>[r]_-{d^1}&
\mathfrak{G}^1 \ar[r]\ar@<1ex>[r]\ar@<-1ex>[r] & \mathfrak{G}^2 
\ar@<1.5ex>[r]\ar@<0.5ex>[r] \ar@<-0.5ex>[r]\ar@<-1.5ex>[r]& \cdots
}
\Bigg)\quad,
\end{flalign}
where as usual we suppressed the codegeneracy maps $s^i$.
\begin{lem}\label{lem:holimcosimplicialgroupoid}
Let $\mathfrak{G}^\bullet \in \Grpd^\Delta$ be any cosimplicial groupoid.
The following groupoid defines a model for the homotopy 
limit $\holim\,\mathfrak{G}^\bullet \in\Grpd$:
\begin{enumerate}[i)]
\item objects are pairs $(x,h)$ consisting of an object $x\in\mathfrak{G}^0$
and a $\mathfrak{G}^1$-morphism $h : d^1(x) \to d^0(x)$, such that
$s^0(h) = \id_x$ and $d^0(h)\circ d^2(h) = d^1(h)$ in $\mathfrak{G}^2$;
\item morphisms $g : (x,h)\to (x^\prime,h^\prime)$ are $\mathfrak{G}^0$-morphisms
$g:x\to x^\prime$, such that the diagram
\begin{flalign}
\xymatrix{
\ar[d]_-{h} d^1(x) \ar[r]^-{d^1(g)} & d^1(x^\prime)\ar[d]^-{h^\prime}\\
d^0(x) \ar[r]_-{d^0(g)} & d^0(x^\prime)
}
\end{flalign}
in $\mathfrak{G}^1$ commutes.
\end{enumerate}
\end{lem}
\begin{rem}
We note that in general the homotopy limit $\holim \,\mathfrak{G}^\bullet$ is 
not weakly equivalent to the ordinary limit $\lim\, \mathfrak{G}^\bullet$. The latter is given by the groupoid
whose objects are all $x\in\mathfrak{G}^0$ satisfying the {\em equality} $d^1(x) = d^0(x)$
and whose morphisms are all $\mathfrak{G}^0$-morphisms $g: x\to x^\prime$ between such objects
that additionally satisfy $d^1(g) = d^0(g)$.
The homotopy limit is weaker than the ordinary limit in the sense that the equality $d^1(x) = d^0(x)$
is promoted to the additional datum of a $\mathfrak{G}^1$-morphism 
$h : d^1(x)\to d^0(x)$ witnessing that $d^1(x)$ and $d^0(x)$
are isomorphic objects in $\mathfrak{G}^1$. This will be crucial in the example below.
\end{rem}

\begin{ex}\label{ex:BGcondescent}
Recall from Example \ref{ex:BGcon} the groupoid $BG^\mathrm{con}(U)\in\Grpd$
of gauge fields with structure group $G$ on a Cartesian space $U$. The assignment
$U\mapsto BG^\mathrm{con}(U)$ is contravariantly functorial on the category $\Cart$
of Cartesian spaces, i.e.\ we have a functor $BG^\mathrm{con} : \Cart^\op\to\Grpd$.
We now shall show that homotopy limits allow us to compute from this information
the groupoid of gauge fields on a general manifold $M$. Let us choose any good open cover $\{ U_i \subseteq M\}$
and form its \v {C}ech nerve
\begin{flalign}\label{eqn:Ubullet}
U_\bullet := \Bigg( \xymatrix@C=1.5em{
\coprod\limits_i U_i  &\ar@<0.5ex>[l]\ar@<-0.5ex>[l] \coprod\limits_{ij} U_{ij}  &\ar@<0ex>[l]\ar@<1ex>[l] \ar@<-1ex>[l] \cdots 
}\Bigg)\quad,
\end{flalign}
where as usual we denote intersections by $U_{i_1\dots i_n} := U_{i_1}\cap \cdots\cap U_{i_n}$.
Using that by hypothesis all non-empty intersections are Cartesian spaces, we
can apply the functor $BG^\mathrm{con} : \Cart^\op\to\Grpd$ and obtain a cosimplicial groupoid
\begin{flalign}\label{eqn:BGconM}
\resizebox{.9\hsize}{!}{$BG^{\mathrm{con}}(U_\bullet)\, :=\,
\Bigg(
\!\!\!\xymatrix@C=1em{
\prod\limits_i BG^{\mathrm{con}}(U_i) \ar@<0.5ex>[r]\ar@<-0.5ex>[r]&
\prod\limits_{ij} BG^{\mathrm{con}}(U_{ij}) \ar[r]\ar@<1ex>[r]\ar@<-1ex>[r] & \cdots
}\!\!\!
\Bigg)$}
\end{flalign}
associated to the cover $\{U_i\subseteq M\}$. Computing the corresponding homotopy limit
$\holim\,BG^{\mathrm{con}}(U_\bullet) \in\Grpd$ according to Lemma \ref{lem:holimcosimplicialgroupoid},
we obtain the groupoid whose
\begin{enumerate}[i)]
\item objects are pairs of families $(\{A_i \in \Omega^1(U_i,\mathfrak{g})\},\{g_{ij} \in C^\infty(U_{ij},G)\} )$,
satisfying 
\begin{enumerate}[i)]
\item $A_j\vert_{U_{ij}} = A_i\vert_{U_{ij}} \triangleleft g_{ij}$, for all $i,j$,
\item $g_{ii} = e$ is the identity of $G$, for all $i$, and 
\item the cocycle condition
$g_{ij}\vert_{U_{ijk}}\,g_{jk}\vert_{U_{ijk}} = g_{ik}\vert_{U_{ijk}}$, for all $i,j,k$;
\end{enumerate}
\item morphisms $(\{A_i \},\{g_{ij} \} ) \to(\{A^\prime_i \},\{g^\prime_{ij} \} ) $ 
are families $\{h_i \in C^\infty(U_i,G)\}$, satisfying 
\begin{enumerate}[i)]
\item $A_i^\prime = A_i \triangleleft h_i$, for all $i$, and 
\item $g^{\prime}_{ij} = h_i^{-1}\vert_{U_{ij}}\,g_{ij}\, h_j\vert_{U_{ij}}$, for all $i,j$.
\end{enumerate}
\end{enumerate}
Observe that this groupoid is precisely the groupoid of gauge fields on $M$,
expressed in terms of \v{C}ech data with respect to the good open cover $\{U_i\subseteq M\}$.
In contrast to this, the ordinary limit $\lim\,BG^{\mathrm{con}}(U_\bullet)\in\Grpd$ is given by the groupoid
whose objects are $1$-forms $A \in \Omega^1(M,\mathfrak{g})$ on $M$ and whose morphisms
are $A\to A\triangleleft h$, with $h\in C^\infty(M,G)$. Note that the latter groupoid describes only 
gauge fields on the trivial principal $G$-bundle $\mathrm{pr}_1: M \times G \to M$, while the correct construction 
by the homotopy limit $\holim\,BG^{\mathrm{con}}(U_\bullet) \in\Grpd$ 
is much richer as it captures all possible principal $G$-bundles on $M$. We also refer to
\cite{Nguyen:2017xie,Dougherty:2017:1189-1201} for a more philosophical perspective 
on gauge fields, groupoids and aspects of richness.
\end{ex}


\subsection{\label{subsec:stacks}The role of stacks}
The groupoid perspective on gauge theories that we have 
introduced in the previous section is incomplete because it neglects
the smooth structure on spaces of gauge fields. We shall now
explain how the concept of stacks resolves this issue.
We refer to \cite{Hollander:0110247} for technical details on the model categorical
approach to $1$-stacks that we review below, and to \cite{Dugger:0205027} 
for analogous developments for  $\infty$-stacks. We also refer to \cite{Lurie:0608040} for an 
$\infty$-categorical approach to stacks
and to \cite{Schreiber:2013pra} for a broad overview, including applications to physics.

The way how stacks formalize smooth structures is more abstract than
the standard approach adopted in differential geometry, which
amounts to endowing a space with an atlas of charts. In order to illustrate the
basic ideas, let us first explain how one can describe manifolds from such a more abstract perspective.
Let $\Man$ be the category of (finite-dimensional) manifolds and smooth maps. 
Instead of describing a manifold $M\in\Man$ by looking for suitable charts,
we shall study the sets $C^\infty(T,M)\in \Set$ of smooth maps from
{\em all} test manifolds $T\in\Man$ into $M$. Note that these sets capture a lot of (in fact, all) 
information about the manifold $M$, for example:
\begin{enumerate}[i)]
\item $C^\infty(\bbR^0,M)$ describes the points in $M$,
\item $C^\infty(\bbR^1, M)$ describes the smooth curves in $M$, and
\item $C^\infty(\bbR^2,M)$ describes the smooth surfaces in $M$, etc.
\end{enumerate}
In particular, the sets $C^\infty(T,M)$ see the smooth structure on $M$. 
Observe that the assignment $T\mapsto C^\infty(T,M)$ is contravariantly
functorial, i.e.\ it defines a functor $C^\infty(-,M) : \Man^\op\to\Set$
to the category of sets. This is called the {\em functor of points} of the manifold $M$.
Because smooth functions between manifolds can be glued, we further observe
that $C^\infty(-,M)\in\mathbf{Sh}(\Man)$ is a sheaf on the site of manifolds $\Man$
with the usual open cover Grothendieck topology.
As a consequence of the Yoneda Lemma, the assignment
\begin{flalign}
\Man \longrightarrow \mathbf{Sh}(\Man)~,~~M\longmapsto C^\infty(-,M)
\end{flalign}
is a fully faithful functor, i.e.\ smooth maps $M\to N$ between
two manifolds can be identified with natural transformations 
$C^\infty(-,M) \to C^\infty(-,N)$ between their functors of points. Finally,
because each manifold admits a good open cover, the category
$\mathbf{Sh}(\Man)$ of sheaves on $\Man$ is equivalent to the category
$\mathbf{Sh}(\Cart)$ of sheaves on the site of Cartesian spaces
$\Cart$ with the good open cover Grothendieck topology.

Summing up, we observed that the category $\Man$
of manifolds can be identified with a full subcategory
\begin{flalign}
\Man \,\subseteq\, \HH_0 := \mathbf{Sh}(\Cart)
\end{flalign}
of the category of $\Set$-valued sheaves on the site of Cartesian spaces $\Cart$.
This means that one can equivalently study manifolds and smooth maps between manifolds
from the perspective of their functors of points. In particular, the smooth structure
of a manifold $M$ is encoded entirely in the smooth mappings $\bbR^n \to M$ from
general test spaces $\bbR^n$, for $n\geq 0$, into $M$. 

Note that the category $\HH_0$ is `vastly bigger' than the category of manifolds $\Man$.
Objects $X\in \HH_0$ that are not (isomorphic to) manifolds should be thought
of as {\em generalized smooth spaces}, where the smooth structure is encoded,
in the spirit of functors of points, by the sets of $\HH_0$-morphisms $\bbR^n\to X$,
for all $n\geq 0$. There are plenty of interesting generalized smooth spaces
that feature in field theory. The following is a small list of concrete examples.
\begin{ex}\label{ex:mappingspaces}
Let $M$ and $N$ be two manifolds, which we regard as objects in $\HH_0$.
Because $\HH_0$ is a Cartesian closed category (even better, it is a {\em topos}),
one can form the internal hom object $[M,N]\in\HH_0$. This is a generalized
smooth space that describes the space of smooth mappings from $M$ to $N$.
Why is that so? To answer this question, let us first look at the points 
$\bbR^0 \to [M,N]$ of this generalized smooth space.
Using that $[M,-]$ is the right adjoint functor of $(-)\times M$,
we can compute the set of points via
\begin{multline}
\Hom_{\HH_0}(\bbR^0,[M,N])\,\cong\, \Hom_{\HH_0}(\bbR^0\times M,N) \\
\,\cong\, \Hom_{\HH_0}(M,N)\,\cong\,C^\infty(M,N)\quad,
\end{multline}
where in the last step we used that the inclusion $\Man\to \HH_0$ is fully faithful.
Thus, the underlying set of points of $[M,N]$ is precisely the set of all smooth
maps from $M$ to $N$. In order to get some feeling for the smooth structure
of $[M,N]$, we note that a similar computation shows that
\begin{flalign}
\Hom_{\HH_0}(\bbR^n,[M,N])\,\cong\, C^\infty(\bbR^n\times M,N)\quad.
\end{flalign}
In particular, the smooth curves $\bbR^1\to [M,N]$ are precisely the smooth
functions $\bbR^1\times M\to N$, which matches the naive expectation
for a smooth structure on a mapping space. We refer to \cite{Benini:2016nqj}
for an application of topos theoretic techniques to non-linear field theories.
\end{ex}

\begin{ex}\label{ex:classifyingforms}
For $p\geq 0$, consider the functor $\Omega^p :$ $\Cart^\op \to \Set\,,~\bbR^n\mapsto  \Omega^p(\bbR^n) $
that assigns $p$-forms to Cartesian spaces. This functor defines a sheaf on $\Cart$
and hence a generalized smooth space $\Omega^p\in\HH_0$. This space is called
the {\em classifying space of $p$-forms} because $\HH_0$-morphisms $M\to \Omega^p$
from a manifold into this space correspond precisely to $p$-forms on $M$. 
Let us provide the relevant argument:
For $M=\bbR^n$, this is a direct consequence of the Yoneda Lemma. For a general
manifold $M$, one chooses any good open cover $\{U_i\subseteq M\}$ 
and uses that $\colim(\coprod_{ij}U_{ij}\rightrightarrows \coprod_i U_i) \to M$
is an isomorphism in $\HH_0$. It follows that a morphism $M\to \Omega^p$ is precisely
a family of $p$-forms $\omega_i\in \Omega^p(U_i)$ satisfying 
$\omega_i\vert_{U_{ij}} = \omega_{j}\vert_{U_{ij}}$, for all $i,j$. The sheaf property
of $p$-forms then implies that this data can be glued to a single $p$-form $\omega\in \Omega^p(M)$
on the manifold $M$. As a side-remark, we would like to mention that
the classifying space $\Omega^p\in\HH_0$ can be used to define a concept 
of $p$-forms on {\em any} generalized smooth space $X\in\HH_0$
in terms of $\HH_0$-morphisms $X\to \Omega^p$. For example, a $p$-form
on the mapping space $[M,N]\in\HH_0$ from Example \ref{ex:mappingspaces}
is simply an $\HH_0$-morphism $[M,N]\to \Omega^p$.
\end{ex}

Let us now turn our attention to stacks. Loosely speaking, 
a stack resembles a generalized smooth space
in the sense above, however with the crucial difference
that its functor of points is valued in $\Grpd$ instead of $\Set$.
(Recall from the previous section that groupoids play a fundamental
role in gauge theory.) In contrast to the strict sheaf condition for generalized
smooth spaces in $\HH_0$, stacks satisfy a weaker {\em homotopy sheaf condition}.
More precisely, we have the following definition \cite{Hollander:0110247}.
\begin{defi}\label{def:stack}
A {\em stack} is a presheaf of groupoids $X : \Cart^\op \to \Grpd$ that satisfies the
homotopy sheaf condition: For each $U\in\Cart$ and good open cover 
$\{U_i\subseteq U\}$, the canonical map
\begin{flalign}
X(U) \stackrel{\sim}{\longrightarrow} \holim\Bigg(
\xymatrix@C=1em{
\prod\limits_i X(U_i) \ar@<0.5ex>[r] \ar@<-0.5ex>[r]&
\prod\limits_{ij} X(U_{ij})  \ar@<1ex>[r] \ar@<-1ex>[r]\ar[r]&
\cdots
}\Bigg)
\end{flalign}
is a weak equivalence in the model category $\Grpd$,
where $\holim$ is the homotopy limit of a cosimplicial groupoid (cf.\ Lemma \ref{lem:holimcosimplicialgroupoid}).
\end{defi}

It was shown in  \cite{Hollander:0110247} that stacks are (the fibrant) objects in a suitable 
model category $\HH_1$. Let us briefly explain this crucial point without going too much into
the details. As a first step, let us consider the category $\mathbf{PSh}(\Cart,\Grpd)$
of groupoid-valued presheaves on $\Cart$. The model structure on $\Grpd$ 
induces the projective model structure on this functor category, i.e.\ 
a morphism $\zeta : X\to Y$ is a fibration (respectively, a weak equivalence)
if all components $\zeta_U : X(U)\to Y(U)$, for $U\in\Cart$, 
are fibrations (respectively, weak equivalences) in $\Grpd$. This model structure
is however not yet quite right, because it does not take into account the Grothendieck
topology on $\Cart$. Given any good open cover $\{U_i\subseteq U\}$
of some $U\in\Cart$, we form its \v{C}ech nerve as in \eqref{eqn:Ubullet},
which defines a simplicial object  $U_\bullet \in \mathbf{PSh}(\Cart,\Grpd)^{\Delta^\op}$
by regarding $U_{i_1\cdots i_n}\in\mathbf{PSh}(\Cart,\Grpd)$ via the Yoneda embedding.
One then defines the model category
\begin{flalign}
\HH_1\,:=\, \mathbf{PSh}(\Cart,\Grpd)_{\mathrm{loc}}
\end{flalign}
by left Bousfield localization of the projective model structure 
at the set of morphisms
\begin{flalign}
\Big\{U \longleftarrow \hocolim U_\bullet~:~ \{U_i\subseteq U\} \text{ good open cover }\Big\}~~,
\end{flalign}
where $\hocolim: \mathbf{PSh}(\Cart,\Grpd)^{\Delta^\op}\to \mathbf{PSh}(\Cart,\Grpd)$ 
is the homotopy colimit with respect to the projective model structure on $\mathbf{PSh}(\Cart,\Grpd)$.
The relationship between stacks and $\HH_1$ is explained in the following proposition, 
which was proven in \cite{Hollander:0110247}.
\begin{propo}\label{prop:stacksfibrant}
Stacks according to Definition \ref{def:stack} are precisely the fibrant
objects in $\HH_1$.
\end{propo}

\begin{ex}
Each manifold $M\in\Man$ defines a stack
by composing its functor of points $C^\infty(-,M): \Cart^\op\to \Set$
with the inclusion $\Set \to \Grpd$. We shall denote the stack corresponding to a manifold
simply by $M\in\HH_1$. More generally, we have an inclusion
$\HH_0\to \HH_1$ of the category of generalized smooth spaces
into $\HH_1$ that takes values in stacks. 
\end{ex}

\begin{ex}\label{ex:homgroupoids}
Recall from Example \ref{ex:BGcon} the presheaf
of groupoids $BG^{\mathrm{con}} : \Cart^\op\to\Grpd$.
From the calculation in Example \ref{ex:BGcondescent}, 
it follows that this defines a stack $BG^{\mathrm{con}}\in\HH_1$,
which is called the {\em classifying stack of principal $G$-bundles 
with connections}. This requires some further explanations.
Let $M$ be a manifold, regarded as an object $M\in \HH_1$.
Computing the naive groupoid $\mathrm{hom}_{\HH_1}(M,BG^{\mathrm{con}})\in\Grpd$ 
of $\HH_1$-morphisms $M\to BG^{\mathrm{con}}$
(see e.g.\ \cite{Benini:2017zjv}), one obtains the groupoid whose objects are
$A\in\Omega^1(M,\mathfrak{g})$ and morphisms are
gauge transformations $A \to A\triangleleft h$, for $h\in C^\infty(M,G)$. 
At first sight that seems very strange, because the latter groupoid 
does not describe non-trivial principal $G$-bundles on $M$ and hence
the name classifying space for $BG^{\mathrm{con}}$ seems unjustified.
So what went wrong? It turns out that computing the groupoids
of $\HH_1$-morphisms $\mathrm{hom}_{\HH_1} : \HH_1^\op\times \HH_1\to \Grpd$
is one of the (many) instances where derived functors are crucial. 
So what went wrong is that we forgot to derive this functor!
Because $BG^{\mathrm{con}}\in\HH_1$ is a fibrant object by 
Proposition \ref{prop:stacksfibrant}, a model for the derived
groupoid of $\HH_1$-morphisms is given by
$\bbR\mathrm{hom}_{\HH_1}(M,BG^{\mathrm{con}})=$\linebreak $ \mathrm{hom}_{\HH_1}(Q M,BG^{\mathrm{con}})$,
where $Q M \to M$ is a cofibrant replacement of the manifold $M$ in $\HH_1$. 
Using as in \cite[Appendix B]{Benini:2017zjv} 
a good open cover $\{U_i\subseteq M\}$ to define a cofibrant replacement
of $M$, one immediately realizes that $\bbR\mathrm{hom}_{\HH_1}(M,BG^{\mathrm{con}})\in\Grpd$
can be computed precisely as the homotopy limit of the cosimplicial groupoid 
displayed in \eqref{eqn:BGconM}. Therefore, recalling Example \ref{ex:BGcondescent},\linebreak
$\bbR\mathrm{hom}_{\HH_1}(M,BG^{\mathrm{con}})\in\Grpd$
is the correct groupoid of all principal $G$-bundles with connections on $M$
(together with their gauge transformations), 
eventually justifying the interpretation of $BG^{\mathrm{con}}$ 
as classifying stack of principal $G$-bundles with connections.
\end{ex}

\begin{ex}\label{ex:GConstack}
Our original aim of this section was to describe a smooth structure on the groupoids 
of gauge fields. This can now be achieved by working within 
the framework of stacks that we discussed above. Let us briefly review how a moduli stack
of gauge fields can be constructed by performing (derived functorial)
constructions in the model category $\HH_1$. As input data,
we choose any manifold $M$ (on which the gauge fields should live)
and any Lie group $G$ (the structure group of the gauge theory).
Recalling the previous two examples, we obtain the two stacks
$M\in\HH_1$ and $BG^{\mathrm{con}}\in\HH_1$. Using that
$\HH_1$ is a Cartesian closed model category (even better, it is a {\em higher topos}),
one can form the derived internal hom object $\bbR[M,BG^{\mathrm{con}}]\in\HH_1$,
which one should interpret similarly to Example \ref{ex:mappingspaces} as a
{\em stack of mappings} from $M$ to $BG^{\mathrm{con}}$. From Example 
\ref{ex:homgroupoids}, we know that the groupoid of
points $\bbR^0 \to \bbR[M,BG^{\mathrm{con}}]$ is the groupoid
of all principal $G$-bundles with connections  on $M$, i.e.\ the 
groupoid of all gauge fields. Unfortunately, the smooth structure
on the mapping stack $\bbR[M,BG^{\mathrm{con}}]$ is not the desired
one, because, as one can show by a direct computation (cf.\ \cite{Benini:2017zjv}),  
a smooth curve $\bbR^1\to \bbR[M,BG^{\mathrm{con}}]$
is given by a principal $G$-bundle with connection on the product manifold 
$\bbR^1\times M$ and not by an $\bbR^1$-parametrized family of principal 
$G$-bundles with connections on $M$. The solution to this issue that
was proposed in \cite{Schreiber:2013pra} and refined in \cite{Benini:2017zjv}
is to perform a {\em differential concretification} of the mapping stack
$\bbR[M,BG^{\mathrm{con}}]$. Loosely speaking, this is a model categorical construction
that `kills off' the bundles and connections on the test spaces $\bbR^n\in\Cart$.
(This can be interpreted in terms of vertical geometry with respect to the 
projection $\mathrm{pr}_1: \bbR^n \times M \to \bbR^n$ onto test spaces.)
As this construction is quite technical, we refer to the original papers for the details.
This defines a new stack $\mathfrak{Con}_G(M)\in\HH_1$, called
the differential concretification of $\bbR[M,BG^{\mathrm{con}}]\in\HH_1$,
which describes our desired moduli stack of principal
$G$-bundles with connections on a manifold $M$. This construction is functorial in
the sense that $\mathfrak{Con}_G : \Man^\op\to \HH_1$ defines a functor
taking values in stacks. (Strictly speaking, this requires the choice of a 
{\em functorial} cofibrant replacement for manifolds, e.g.\ the one 
in \cite[Appendix B]{Benini:2017zjv}.) As a last remark, we would like to add
that moduli stacks of solutions to, e.g., the non-Abelian Yang-Mills equation
or the Chern-Simons equation can also be constructed from such a perspective.
See  \cite{Benini:2017zjv} for the details on Yang-Mills theory
and \cite{Fiorenza:2013jz} for Chern-Simons theory.
\end{ex}

\begin{rem}\label{rem:inftystacks}
We conclude this section by briefly commenting on how to describe
$\infty$-stacks from a model categorical perspective. See e.g.\ 
\cite{Dugger:0205027,Dugger:0007070} for the details. Let us recall
that an explicit model for the $\infty$-category $\infty\Grpd$
is given by endowing the category $\sSet = \Set^{\Delta^\op}$ of simplicial sets
with the usual Kan-Quillen model structure. The fibrant objects in
this model category are the Kan complexes, which are a model for $\infty$-groupoids.
Instead of $\Grpd$-valued presheaves on $\Cart$, the description
of $\infty$-stacks starts from the category $\mathbf{PSh}(\Cart,\sSet)$ of
presheaves with values in $\sSet$. The left Bousfield localization of the projective
model structure at all hypercovers defines the model category
\begin{flalign}
\HH_\infty\,:=\,\mathbf{PSh}(\Cart,\sSet)_\mathrm{loc}\quad.
\end{flalign}
$\infty$-stacks are then by definition the fibrant objects in $\HH_\infty$.
Similarly to Definition \ref{def:stack} and Proposition \ref{prop:stacksfibrant}, 
$\infty$-stacks can be characterized  by a suitable homotopy sheaf condition 
with respect to hypercovers. Finally, the inclusions from Remark 
\ref{rem:setgrpdinftygrpd} generalize to
\begin{flalign}
\HH_0 \hookrightarrow \HH_1\hookrightarrow\HH_2\hookrightarrow\cdots\hookrightarrow\HH_\infty\quad,
\end{flalign}
which means that all generalized smooth spaces, stacks, $2$-stacks,\dots,
can be regarded as particular examples of $\infty$-stacks.
\end{rem}


\subsection{\label{subsec:cochainalgebras}Smooth cochain algebras on stacks}
In the previous sections we have seen that higher structures
are crucial for the description of `spaces' of gauge fields, which
are in fact higher categorical spaces called stacks. Thinking ahead towards
QFT, which requires a concept of observable algebras,
we would like to explain what it means to form `function algebras' 
on stacks. Our statements below are formulated for the more general
case of $\infty$-stacks in $\HH_\infty$, because this does not lead
to any further complications compared to the case of $1$-stacks in $\HH_1$.

Before explaining our concept of smooth cochain algebras on 
$\infty$-stacks, we would like to start with the related, but simpler,
case of cochain algebras on simplicial sets. See e.g.\ \cite{Benini:2018oeh} 
for a more extensive review. In the following $k$ will be a field of characteristic $0$, 
e.g.\ $k=\bbR$ or $k=\bbC$, and $\Ch(k)$ the symmetric monoidal model category 
of (possibly unbounded) chain complexes of $k$-vector spaces, see e.g.\ \cite{hovey2007model}.
Let us recall that in this model structure a morphism $f : V\to W$ between two chain complexes is
\begin{enumerate}[i)]
\item a weak equivalence if it is a quasi-isomorphism, i.e.\ it induces an isomorphism
$H_\bullet(f) : H_\bullet(V)\to H_\bullet (W)$ in homology;
\item a fibration if it is degree-wise surjective;
\item a cofibration if it has the left lifting property with respect
to all acyclic fibrations.
\end{enumerate}
Recall that to every simplicial set $S\in\sSet$, one can associate
the chain complex $N_\ast(S,k)\in\Ch(k)$ of {\em normalized $k$-valued 
chains} on $S$. The functor $N_\ast(-,k) : \sSet \to \Ch(k)$
is the left adjoint of a Quillen adjunction between the model category 
$\sSet$, with the Kan-Quillen model structure, and the model category $\Ch(k)$.
Composing $N_\ast(-,k)$ with the internal hom functor $[-,k] : \Ch(k)\to\Ch(k)^\op$
for chain complexes, which is also a left Quillen functor, defines a left
Quillen functor $N^\ast(-,k) : \sSet\to\Ch(k)^\op$ that assigns to a
simplicial set its {\em normalized $k$-valued cochains}. By \cite{Berger:0109158},
the latter are canonically $E_\infty$-algebras, i.e.\ homotopy-coherently commutative differential graded 
algebras.  (We refer to Section \ref{sec:homotopy} for more details on homotopy algebras over operads.)
Summing up, we obtained a left Quillen functor
\begin{flalign}
N^\ast(-,k) \,:\,\sSet\longrightarrow \Alg_{\E_\infty}(\Ch(k))^\op
\end{flalign}
that assigns to each simplicial set $S\in\sSet$ its normalized cochain algebra 
$N^\ast(S,k)\in\Alg_{\E_\infty}(\Ch(k))$. Because all simplicial sets
are cofibrant in the Kan-Quillen model structure, this functor preserves
weak equivalences and does not have to be derived.
\begin{ex}
Every set $S\in\Set$ can be regarded as a constant simplicial
set that we also denote by $S\in\sSet$. The normalized cochain algebra
in this case is a chain complex concentrated in degree $0$ with trivial differential,
i.e.\ it is just a vector space.
A concrete calculation shows that $N^\ast(S,k) = \mathrm{Map}(S,k)$ 
is the usual commutative algebra of $k$-valued functions on the set $S$.
\end{ex}
\begin{ex}
More interestingly, let $\mathfrak{G}\in\Grpd$ be a groupoid
and consider its nerve $B\mathfrak{G}\in\sSet$. Then the normalized
cochain algebra $N^\ast(B\mathfrak{G},k)$ is precisely the usual 
$k$-valued groupoid cohomology dg-algebra, see e.g.\ \cite{Crainic:2003:681-721}.
Note that this is in general not a strictly commutative dg-algebra, but an $E_\infty$-algebra.
Thinking of $\mathfrak{G}\in\Grpd$ as a groupoid of gauge fields, 
the dg-algebra $N^\ast(B\mathfrak{G},k)$ describes both functions
of gauge fields (in degree $0$) and functions of ghost fields
(in homological degree $<0$, or in cohomological degree $>0$).
In fact, $N^\ast(B\mathfrak{G},k)$ is a groupoid version of the usual Chevalley-Eilenberg
dg-algebra from Lie algebroid cohomology.
\end{ex}

The construction of cochain algebras above can be generalized to the case
of $\infty$-stacks. We shall provide a brief sketch and refer to \cite{Benini:2018oeh} for the technical 
details. Let us consider for the moment the case where all presheaf categories
(with values in model categories) are endowed with the projective model structures.
Applying the normalized chain functor $N_\ast(-,k) : \sSet\to \Ch(k)$ object-wise
on presheaves defines a left Quillen functor that we denote with abuse of notation
by the same symbol
\begin{flalign}\label{eqn:Nchaininfty}
N_\ast(-,k) : \mathbf{PSh}(\Cart,\sSet)\longrightarrow \mathbf{PSh}(\Cart,\Ch(k))\quad.
\end{flalign}
Concretely, given a presheaf of simplicial sets $X : \Cart^\op \to \sSet$,
then $N_\ast(X,k) : \Cart^\op\to\Ch(k)$ is the presheaf of chain
complexes defined by $N_\ast(X,k)(U) := N_\ast(X(U),k)\in\Ch(k)$, for all $U\in\Cart$.
Assuming in the following that $k=\bbR$ or $k=\bbC$, we can promote
$k$ to an object $\und{k} \in \mathbf{PSh}(\Cart,\Ch(k))$
by setting $\und{k}(U) := C^\infty(U,k)$ to be the vector space of $k$-valued
smooth functions, for all $U\in\Cart$. Because $\mathbf{PSh}(\Cart,\Ch(k))$ is enriched
over $\Ch(k)$ and $\und{k}$ is a fibrant object, there is a left Quillen functor 
\begin{flalign}\label{eqn:hominfty}
[-,\und{k}]^\infty \,:\, \mathbf{PSh}(\Cart,\Ch(k))\longrightarrow\Ch(k)^\op
\end{flalign}
that  assigns chain complexes of morphisms.
This functor should be understood as taking smooth $k$-valued functions. Concretely,
given any object $V\in \mathbf{PSh}(\Cart,\Ch(k))$, one has an explicit description 
by an end
\begin{flalign}
[V,\und{k}]^\infty  = \int_{U\in\Cart^\op} \big[V(U),\und{k}(U)\big]\quad,
\end{flalign}
where $[-,-]$ denotes the internal hom functor in $\Ch(k)$.
Composing \eqref{eqn:Nchaininfty} and \eqref{eqn:hominfty}
defines a left Quillen functor
\begin{flalign}\label{eqn:Ncochaininfty}
N^{\infty \ast}(-,k)\,:\, \mathbf{PSh}(\Cart,\sSet)\longrightarrow \Alg_{\E_{\infty}}(\Ch(k))^\op
\end{flalign}
that assigns $E_\infty$-algebras because of \cite{Berger:0109158}.
We shall call this the {\em smooth normalized cochain algebra functor}.
\begin{propo}
The left derived functor $\bbL N^{\infty \ast}(-,k)$ of \eqref{eqn:Ncochaininfty}
restricts to a homotopical functor on the full 
subcategory $\mathbf{St}_\infty \subseteq \HH_\infty$ of $\infty$-stacks 
(the fibrant objects in $\HH_\infty$), i.e.\
\begin{flalign}\label{eqn:NcochaininftySTACKS}
\bbL N^{\infty \ast}(-,k)\,:\, \mathbf{St}_\infty \longrightarrow \Alg_{\E_{\infty}}(\Ch(k))^\op
\end{flalign}
is a functor that preserves weak equivalences between $\infty$-stacks.
\end{propo}

\begin{rem}
In contrast to normalized cochain algebras on simplicial sets, the functor \eqref{eqn:Ncochaininfty}
must be derived because not every presheaf of simplicial sets is a cofibrant object.
A relatively concrete model for cofibrant replacement in this case
is given by Dugger in \cite{Dugger:0007070}.
\end{rem}

\begin{ex}\label{ex:evidencehocostructures}
Let $\CC$ be a category (of space-times) 
and $\mathfrak{F} : \CC^\op \to \mathbf{St}_\infty \subseteq \HH_\infty$
a functor that assigns  to each space-time $c\in\CC$ an $\infty$-stack
$\mathfrak{F}(c)$ of gauge fields on $c$. For example,
this could be the functor $\mathfrak{Con}_G : \Man^\op \to \mathbf{St}_\infty$ 
assigning the moduli stacks of principal $G$-bundles with connections
(cf.\  Example \ref{ex:GConstack}) or the functor
$\mathfrak{YM}_G : \Loc^\op\to \mathbf{St}_\infty$ assigning
the moduli stacks of solutions of the Yang-Mills equation 
(cf.\ \cite{Benini:2017zjv}). Applying \eqref{eqn:NcochaininftySTACKS}
defines for each space-time $c\in\CC$ an $E_\infty$-algebra
\begin{flalign}
\AAA(c) \,:=\, \bbL N^{\infty \ast}\big(\mathfrak{F}(c),k\big)\in \Alg_{\E_\infty}(\Ch(k))
\end{flalign}
that one should interpret as a classical observable algebra for
the gauge fields on $c$. Functoriality of this construction implies that
$\AAA : \CC\to\Alg_{\E_\infty}(\Ch(k))$ is a covariant functor
on $\CC$, which is very similar to the structures considered in AQFT, cf.\ Section \ref{sec:operads}.
There are however two main differences:
1.)~The observable algebras that are assigned here are dg-algebras, i.e.\ associative and unital
algebras in chain complexes, not ordinary algebras in vector spaces.
The higher structures in the $\infty$-stacks of gauge fields are 
represented by higher homology groups of these dg-algebras. 
2.)~Even though no quantization happened so far, the observable algebras are not strictly commutative,
but commutative up to coherent homotopies. Therefore, in order to 
understand gauge theories in AQFT, one is naturally lead to consider homotopy-coherent algebraic structures.
This will be formalized in Section \ref{sec:homotopy}.
\end{ex}

\begin{problem}
For applications to gauge theory, the $\infty$-stacks $X$ typically
carry a Poisson structure (or symplectic structure) which is determined 
by the action functional. It is currently unclear to us how one can 
construct (in a homotopically meaningful way) a Poisson bracket on
the $E_\infty$-algebra $\bbL N^{\infty \ast}(X,k)$. It is even 
more unclear to us how one can quantize (in a homotopically meaningful way)
such homotopy-coherent versions of Poisson dg-algebras, which is required
for constructing examples of quantum gauge theories.
Quite recently there have been impressive developments in derived
algebraic geometry \cite{Pantev:1111.3209,Calaque:1506.03699} that focus on related questions
and we hope to see some fruitful interplay with this discipline  in the future.
\end{problem}


\subsection{\label{subsec:derivedgeometry}Derived geometry of linear gauge fields}
There is a second kind of higher structures in (gauge) field theory
that has a different origin than the groupoid (or $\infty$-groupoid)
structures discussed in the previous sections. Summarizing the latter in a single 
sentence, groupoids and stacks become important whenever one divides 
out gauge symmetries that usually do not act freely on the gauge fields.
(Recall the role of stabilizers at the beginning of Section \ref{subsec:groupoids}.)
In a (gauge) field theory one typically starts from a space or stack $\mathfrak{F}$ of fields 
together with an action functional $S : \mathfrak{F}\to \bbR$. The aim is then
to describe the space of solutions of the corresponding Euler-Lagrange equations,
which can be obtained from the following construction: First, one considers the
variation of the action $S$, which defines a section $\dd S : \mathfrak{F} \to T^\ast \mathfrak{F}$ 
of the cotangent bundle over the space (or stack) of fields. Informally, the 
space of solutions is the `subspace' $\mathfrak{Sol} \subseteq \mathfrak{F}$ on which 
$\dd S: \mathfrak{F}\to T^\ast \mathfrak{F}$
coincides with zero section $0 : \mathfrak{F}\to T^\ast \mathfrak{F}$.
This is formalized by forming the (homotopy) pullback 
\begin{flalign}\label{eqn:criticallocus}
\xymatrix{
\ar@{-->}[d] \mathfrak{Sol} \ar@{-->}[r] & \ar@{}[dl]_-{h~~~~~} \mathfrak{F} \ar[d]\ar[d]^-{\dd S}\\
\mathfrak{F} \ar[r]_-{0} & T^\ast\mathfrak{F}
}
\end{flalign}
in the appropriate (model) category of spaces or stacks.
Geometrically, one can interpret this construction as computing the intersection
of $\dd S$ with the zero section $0$. 

For a generic action functional $S$, the intersection in \eqref{eqn:criticallocus} will be far 
away from being transversal and hence the space (or stack) $\mathfrak{Sol}$
can be badly behaved. Furthermore, it ignores more refined information about the intersection
problem, such as the multiplicities of multiple intersections.
A solution to these problems is proposed by derived algebraic geometry \cite{Pantev:1111.3209,Calaque:1506.03699},
where a more refined concept of spaces, called {\em derived $\infty$-stacks}, is developed. 

Let us explain very briefly the basic idea behind derived $\infty$-stacks, without going into
any technical details. Recall from Section \ref{subsec:stacks} that an $\infty$-stack
is described by its functor of points $X : \Cart^\op\to \sSet$ that assigns
to each Cartesian space $U\in\Cart$ the $\infty$-groupoid $X(U)$ of points of
shape $U$ in $X$. Recall that the latter encode both the gauge fields and the (higher) gauge symmetries. A derived
$\infty$-stack is a more refined concept that is described by a
functor of points of the form $X : \mathbf{cCart}^\op \to \sSet$, where
$\mathbf{cCart} = \Cart^{\Delta}$ are {\em cosimplicial test spaces}.
(In algebraic geometry, these are described by the opposite category
of simplicial commutative $k$-algebras $\mathbf{sCAlg}_k$, cf.\  \cite{Pantev:1111.3209,Calaque:1506.03699}.)
Note that there are two opposite degrees appearing in a derived $\infty$-stack:
The `stacky' simplicial degree in the target category $\sSet$ and the `derived'
cosimplicial degree in the source category $\mathbf{cCart}$. Morally speaking, 
the former encodes refined aspects of gauge symmetries and the latter 
encodes refined aspects of intersections. We will later see that these two different
degrees are related to ghost fields and anti fields in the BRST/BV formalism.

Working with derived $\infty$-stacks is very hard. In particular, we are not yet
able to describe physically interesting examples of solutions spaces of gauge 
theories within this approach. (Toy-models of such are discussed in \cite{Pantev:1111.3209,Calaque:1506.03699}.)
In what follows we shall focus on a certain approximation of the structures
appearing in derived algebraic geometry, which however encodes some of the
crucial features of this approach. Let us motivate this approximation.
As mentioned above, derived $\infty$-stacks come with two degrees, 
`stacky' degrees in $\sSet$ and `derived' degrees in $\mathbf{cCart}$.
If one restricts to linear spaces and linear maps between them, the Dold-Kan 
correspondence allows us to describe the `stacky' degrees by
non-negatively graded chain complexes $\Ch_{\geq 0}(k)$ and the
`derived' degrees by non-positively graded chain complexes $\Ch_{\leq 0}(k)$.
Our working assumption below is thus that unbounded chain complexes
$\Ch(k)$ capture linear features of derived $\infty$-stacks.
A similar perspective is taken in
the work of Costello and Gwilliam \cite{costello2016factorization}.

Let us now focus on a very simple example to illustrate the main features
of derived $\infty$-stacks in field theory. In what follows we fix 
$k=\bbR$ and take any oriented and time-oriented globally hyperbolic 
Lorentzian manifold $M\in\Loc$, interpreted as space-time. 
We consider Abelian gauge fields with structure group 
$G=\bbR$. Because there are no non-trivial principal $\bbR$-bundles,
the groupoid of gauge fields (cf.\ Example \ref{ex:BGcon}) on $M$ is given by
\begin{flalign}
BG^\mathrm{con}(M) \,=\,\begin{cases}
\text{Obj:} &  A\in \Omega^1(M)\\
\text{Mor:} & A\stackrel{\epsilon}{\longrightarrow} A + \dd \epsilon \\
& \text{with }  \epsilon \in C^\infty(M)
\end{cases}\quad\quad.
\end{flalign}
Following \cite{Benini:2015hta}, one easily computes the nerve of this groupoid 
and, after applying the Dold-Kan correspondence to the resulting 
simplicial vector space, obtains the chain complex
\begin{flalign}
\mathfrak{F} = \Big(
\xymatrix@C=1em{
\stackrel{(0)}{\Omega^1(M)} & \ar[l]_-{\dd} \stackrel{(1)}{\Omega^0(M)}
}
\Big)\quad,
\end{flalign}
where we indicated  in round brackets the homological degrees
and identified functions with $0$-forms $\Omega^0(M)=C^\infty(M)$.
As expected, this chain complex has only `stacky' positive degrees and 
no `derived' negative degrees. 

As action functional we take the ordinary Abelian Yang-Mills action
on the space-time $M$, i.e.\ $S = \int_M \frac{1}{2} \dd A \wedge \ast \dd A$,
where $\ast$ is the Hodge operator. A naive variation of this action
leads to the Abelian Yang-Mills equation $\delta \dd A = 0$,
where $\delta$ is the codifferential, i.e.\ the formal adjoint of $\dd$
with respect to the  inner product 
$\langle \omega,\lambda\rangle = \int_M \omega \wedge \ast \lambda$ on $p$-forms.
We would like to describe the space of solutions of the Abelian Yang-Mills
equation from the perspective of derived geometry by computing a 
homotopy pullback as in \eqref{eqn:criticallocus}, which is also called 
the {\em derived critical locus}.  For this we have to introduce a cotangent bundle $T^\ast \mathfrak{F}$
over the space of fields $\mathfrak{F}$. 
In line with our working assumption that $\mathfrak{F} \in \Ch(k)$
models a `linear space' (in the sense of being a chain complex of $k$-vector spaces),
it is reasonable to {\em define} the cotangent bundle
\begin{flalign}
T^\ast \mathfrak{F} \,:=\, \mathfrak{F}\times \mathfrak{F}^\ast\,\in\Ch(k)
\end{flalign}
as the product in $\Ch(k)$ of the chain complex of fields and an appropriate 
choice of linear dual of that. Choosing the `smooth dual'
\begin{flalign}\label{eqn:Fdual}
\mathfrak{F}^\ast \,:=\,  \Big(
\xymatrix@C=1em{
 \stackrel{(-1)}{\Omega^0(M)} & \ar[l]_-{-\delta} \stackrel{(0)}{\Omega^1(M)} 
}
\Big)\quad,
\end{flalign}
one obtains
\begin{flalign}
\resizebox{.87\hsize}{!}{$T^\ast \mathfrak{F} \,=\,
\Big(
\xymatrix@C=1em{
 \!\!\!\!\stackrel{(-1)}{\Omega^0(M)} & \ar[l]_-{-\delta\pi_2} \stackrel{(0)}{\Omega^1(M) \times \Omega^1(M)} & \ar[l]_-{\iota_1 \dd}  \stackrel{(1)}{\Omega^0(M)} \!\!\!\!
}
\Big)~~,$}
\end{flalign}
where $\iota_1 : \Omega^1(M)\to \Omega^1(M) \oplus \Omega^1(M) =  \Omega^1(M) \times \Omega^1(M)$
is the inclusion of the first factor and $\pi_2 : \Omega^1(M) \times \Omega^1(M)\to \Omega^1(M)$
the projection on the second factor.
\begin{problem}\label{problem:duals}
Note that the `smooth dual' $\mathfrak{F}^\ast$ we have chosen 
in \eqref{eqn:Fdual} is {\em not} obtained via the categorical concept 
of dual chain complexes. In fact, the dualizable objects in $\Ch(k)$ are perfect
complexes, i.e.\ chain complexes that are quasi-isomorphic
to a bounded chain complex of finite-dimensional $k$-vector spaces,
however $\mathfrak{F}\in\Ch(k)$ is clearly not perfect.
As a consequence, it is presently not clear to us if the construction
of the cotangent bundle $T^\ast \mathfrak{F}$ preserves weak equivalences
in $\Ch(k)$. According to our best knowledge, it is an open problem
how to formalize a model categorical (or higher categorical) concept
of `smooth duals' as in \eqref{eqn:Fdual}.
\end{problem}

The variation of the action $S$ defines a $\Ch(k)$-mor-\linebreak phism
$\dd S : \mathfrak{F}\to T^\ast \mathfrak{F}$
that is concretely given by
\begin{flalign}
\xymatrix@C=2em@R=1.5em{
\ar[d]_-{0} 0 & \ar[l]_-{0} \Omega^1(M) \ar[d]_-{(\id,\delta\dd)}& \ar[l]_-{\dd} \Omega^0(M)\ar[d]^-{\id}\\
\Omega^0(M) &\ar[l]^-{-\delta \pi_2} \Omega^1(M)\times\Omega^1(M) &\ar[l]^-{\iota_1\dd} \Omega^0(M)
}
\end{flalign}
Moreover, the zero section $0 : \mathfrak{F}\to T^\ast\mathfrak{F}$
is the $\Ch(k)$-morphism given by
\begin{flalign}\label{eqn:zerosection}
\xymatrix@C=2em@R=1.5em{
\ar[d]_-{0} 0 & \ar[l]_-{0} \Omega^1(M) \ar[d]_-{(\id,0)}& \ar[l]_-{\dd} \Omega^0(M)\ar[d]^-{\id}\\
\Omega^0(M) &\ar[l]^-{-\delta \pi_2} \Omega^1(M)\times\Omega^1(M) &\ar[l]^-{\iota_1\dd} \Omega^0(M)
}
\end{flalign}
\begin{propo}\label{propo:criticallocus}
Consider as above Abelian Yang-Mills theory with structure group $G=\bbR$
on a space-time $M\in\Loc$. A model for the corresponding homotopy pullback 
\eqref{eqn:criticallocus} in $\Ch(k)$ is given by
\begin{flalign}\label{eqn:YMcomplex}
\resizebox{.87\hsize}{!}{$\mathfrak{Sol} = 
\Big(
\xymatrix@C=1em{
 \!\!\!\!\stackrel{(-2)}{\Omega^0(M)}& \ar[l]_-{\delta } \stackrel{(-1)}{\Omega^1(M)} & \ar[l]_-{\delta\dd} \stackrel{(0)}{\Omega^1(M)} & \ar[l]_-{\dd}  \stackrel{(1)}{\Omega^0(M)} \!\!\!\!
}
\Big)~~.$}
\end{flalign}
\end{propo}
\begin{proof}
By \cite[Corollary 13.1.3]{Hirschhorn:2009:47-69}, $\Ch(k)$ is a {\em right proper}
model category because each object is fibrant. As a consequence
of \cite[Corollary 13.3.8]{Hirschhorn:2009:47-69}, one can compute the 
homotopy pullback \eqref{eqn:criticallocus} in terms of
an {\em ordinary} pullback if we replace 
the zero section $0: \mathfrak{F}\to T^\ast\mathfrak{F}$
by a weakly equivalent fibration. Because the zero section is
the Cartesian product of the $\Ch(k)$-morphisms 
$\id : \mathfrak{F}\to \mathfrak{F}$ and $0 : 0\to \mathfrak{F}^\ast$,
the problem reduces to finding a fibration
that is weakly equivalent to the zero map $0 : 0\to \mathfrak{F}^\ast$.

For this let us introduce the chain complex
\begin{flalign}
D \,:=\, \Big(
\xymatrix@C=1.5em{
\stackrel{(-1)}{k} & \ar[l]_-{\id} \stackrel{(0)}{k}
}
\Big)\,\in\Ch(k)\quad,
\end{flalign}
which allows us to factorize the unique $\Ch(k)$-morphism $0\to k$
into an acyclic cofibration $0\to D$ followed by a fibration
$D\to k$. (Concretely, the latter map is given by $\id:k\to k$ in degree $0$
and $0:k\to 0$ in degree $-1$.) Taking the tensor product
$(-)\otimes \mathfrak{F}^\ast$ of the factorization $0\to D\to k$
yields a factorization 
\begin{flalign}
\xymatrix@C=2em{
0 \ar[r] & D\otimes \mathfrak{F}^\ast \ar[r]^-{p} & \mathfrak{F}^\ast
}
\end{flalign}
of the zero map $0 : 0\to \mathfrak{F}^\ast$ into a weak equivalence
$0\to D\otimes\mathfrak{F}^\ast$ followed by a fibration
$p: D\otimes \mathfrak{F}^\ast\to\mathfrak{F}^\ast$.
Hence, we have constructed a replacement
\begin{flalign}
\widetilde{0} :=\id\times p  \,: \, \widetilde{T}^\ast \mathfrak{F} := 
\mathfrak{F}\times (D\otimes \mathfrak{F}^\ast) \,\longrightarrow\,
\mathfrak{F}\times\mathfrak{F}^\ast =T^\ast \mathfrak{F}
\end{flalign}
of the zero section $0: \mathfrak{F}\to T^\ast\mathfrak{F}$
by a weakly equivalent fibration. 
 
Let us now compute explicitly the {\em ordinary} pullback 
\begin{flalign}\label{eqn:ordinarypullback}
\xymatrix{
\ar@{-->}[d] \mathfrak{Sol} \ar@{-->}[r] &  \mathfrak{F} \ar[d]\ar[d]^-{\dd S}\\
\widetilde{T}^\ast \mathfrak{F} \ar[r]_-{\widetilde{0}} & T^\ast\mathfrak{F}
}
\end{flalign}
in $\Ch(k)$, which provides a model for the desired 
homotopy pullback \eqref{eqn:criticallocus}. For this we use that 
the chain complex $D\otimes \mathfrak{F}^\ast$ is concretely given by
\begin{flalign}
\resizebox{.87\hsize}{!}{$D\otimes \mathfrak{F}^\ast
= \Big( \!\!\!\!\xymatrix@C=1em{
\stackrel{(-2)}{\Omega^0(M)} & \ar[l]_-{\delta\pi_1+\pi_2} \stackrel{(-1)}{\Omega^1(M)\times \Omega^0(M)}
& \ar[l]_-{(\id,-\delta)} \stackrel{(0)}{\Omega^1(M)}
} \!\!\!\!
\Big)$}
\end{flalign}
and that the $\Ch(k)$-morphism $p : D\otimes \mathfrak{F}^\ast\to\mathfrak{F}^\ast$
reads as
\begin{flalign}
\xymatrix@C=2em@R=1.5em{
\ar[d]_-{0}\Omega^0(M) & \ar[l]_-{\delta\pi_1+\pi_2}\ar[d]_-{\pi_2} \Omega^1(M)\times \Omega^0(M)
& \ar[l]_-{(\id,-\delta)} \Omega^1(M)\ar[d]^-{\id}\\
0&\ar[l]^-{0}\Omega^0(M)&\ar[l]^-{-\delta}\Omega^1(M)
}
\end{flalign}
One easily computes the lower horizontal arrow in \eqref{eqn:ordinarypullback} 
and confirms that the pullback is given by \eqref{eqn:YMcomplex}.
\end{proof}

\begin{rem}
The following remarks are in order:
\begin{enumerate}[i)]
\item As a graded vector space, the model $\mathfrak{Sol}$ 
for the derived critical locus established 
in Proposition \ref{propo:criticallocus}
agrees with the shifted cotangent bundle over $\mathfrak{F}$. 
Hence, the above proof of Proposition \ref{propo:criticallocus} 
provides a homological explanation for the appearance 
of shifted cotangent bundles in the calculation of derived critical loci. 

\item The chain complex $\mathfrak{Sol}$ in 
\eqref{eqn:YMcomplex} has both `stacky' positive degrees
and `derived' negative degrees. The different components have a 
physical interpretation in terms of the BRST/BV formalism:
Fields in degree $0$ are called {\em gauge fields} $A\in\Omega^1(M)$ 
and fields in degree $1$ {\em ghost fields} $c\in\Omega^0(M)$. The fields
in negative degrees are called {\em anti fields} $A^\ddagger\in\Omega^1(M)$ 
and $c^\ddagger \in \Omega^0(M)$.
\end{enumerate}
\end{rem}

Let us now compute the homologies of the chain complex 
\eqref{eqn:YMcomplex}. In degree $1$ we obtain the zeroth 
de Rham cohomology of the space-time $M$ 
\begin{flalign}
H_1(\mathfrak{Sol}) \,\cong\, H^0_{\mathrm{dR}}(M) \,\cong\, \bbR^{\pi_0(M)}\quad,
\end{flalign}
where $\pi_0(M)$ denotes the set of connected components of $M$.
Note that this is always non-zero and it agrees with our computation
of $\pi_1(BG^{\mathrm{con}}(U),A)$ for the groupoid of gauge fields in Example \ref{ex:BGcon}.
In degree $0$ we obtain the space of gauge equivalences
classes of solutions of the Abelian Yang-Mills equation
\begin{flalign}
H_0(\mathfrak{Sol}) \,\cong\,\frac{\big\{A\in \Omega^1(M) \,:\, \delta\dd A =0\big\}}{\dd\Omega^0(M)}\quad.
\end{flalign}
Working out the homology in degree $-1$ is slightly more complicated.
Because $M$ is by hypothesis a {\em globally hyperbolic} space-time,
one can show that the inhomogeneous Abelian Yang-Mills equation
$\delta\dd A = j$, for $j\in \Omega^1(M)$, has a solution $A\in\Omega^1(M)$
if and only if $j$ is $\delta$-exact, i.e.\ $j = \delta \eta$ for some $\eta\in\Omega^2(M)$.
(Hint: Apply a gauge transformation to $A$ to fulfill the Lorenz gauge condition $\delta A =0$
and then use standard techniques from the theory of wave equations \cite{Baer:0806.1036}.)
Hence,
\begin{flalign}
\mathrm{Im}\big(\delta\dd : \Omega^1(M)\to \Omega^1(M)\big) \,=\, \delta \Omega^2(M)
\end{flalign}
is the space of $\delta$-exact $1$-forms and the degree $-1$ homology
of \eqref{eqn:YMcomplex} is
\begin{flalign}
H_{-1}(\mathfrak{Sol})\,\cong\, H^1_\delta(M)\, \cong\, H^{m-1}_{\mathrm{dR}}(M)\quad,
\end{flalign}
where $H^\bullet_\delta$ is the cohomology of the codifferential. (Recall that $m$ is the dimension
of $M$.) Finally, the degree $-2$ homology of \eqref{eqn:YMcomplex} is trivial
\begin{flalign}
H_{-2}(\mathfrak{Sol})\,\cong\, H^0_\delta(M) \,\cong\, H^m_{\mathrm{dR}}(M)\,\cong\, 0\quad,
\end{flalign}
because the underlying manifold of a globally hyperbolic space-time is diffeomorphic to
$M\cong \bbR\times \Sigma$, with $\Sigma$ an $m-1$-dimensional manifold. 
Note that this calculation shows that the $\Ch(k)$-morphism
\begin{flalign}
\xymatrix@C=1.5em@R=1.5em{
\Omega^0(M)& \ar[l]_-{\delta } \Omega^1(M) & \ar[l]_-{\delta\dd} \Omega^1(M) & \ar[l]_-{\dd}  \Omega^0(M)\\
\ar[u]^-{0} 0 &\ar[u]^-{\subseteq} \ar[l]^-{0} \Omega^1_\delta(M) & \ar[u]^-{\id} \ar[l]^-{\delta\dd } \Omega^1(M) &\ar[u]_-{\id}\ar[l]^-{\dd} \Omega^0(M)
}
\end{flalign}
is a quasi-isomorphism, where by $\Omega^1_\delta(M)$ we denoted the $\delta$-closed $1$-forms.
\begin{cor}
For every globally hyperbolic Lorentzian space-time $M\in\Loc$,
the solution complex $\mathfrak{Sol}$ in \eqref{eqn:YMcomplex} is weakly equivalent
to the smaller chain complex
\begin{flalign}\label{eqn:smallcomplex}
\widetilde{\mathfrak{Sol}} \,:=\,
\Big(
\xymatrix@C=1em{
\stackrel{(-1)}{\Omega^1_\delta(M)} & \ar[l]_-{\delta\dd} \stackrel{(0)}{\Omega^1(M)} & \ar[l]_-{\dd}  \stackrel{(1)}{\Omega^0(M)} 
}
\Big)\quad.
\end{flalign}
\end{cor}

We would like to conclude this section with a comment on the time-slice axiom
in the present setting. Let us take any Cauchy morphism
$f:M\to N$ in $\Loc$, i.e.\ $f(M)\subseteq N$ contains a Cauchy surface of $N$. 
Pullback of differential forms defines a
$\Ch(k)$-morphism 
\begin{flalign}
f^\ast \,:\, \mathfrak{Sol}(N)\longrightarrow \mathfrak{Sol}(M)
\end{flalign}
from the chain complex of solutions on $N$ to the one on $M$.
One easily observes that $f^\ast$ is not an isomorphism in $\Ch(k)$,
but rather a quasi-isomorphism. (The same statements
hold true for the smaller complex in \eqref{eqn:smallcomplex}.)
Thinking ahead towards homotopical AQFT, this means that one should not
expect a strict time-slice axiom to hold true in such gauge theoretic 
examples, but rather a homotopy-coherent generalization of it.


\section{\label{sec:homotopy}Homotopy theory of AQFTs}
We develop a general framework for AQFTs with values in the symmetric monoidal
model category $\Ch(k)$ of chain complexes of $k$-modules.
This is motivated by the higher structures arising in gauge theory 
that we explained in Section \ref{sec:gaugetheory} above. Concrete examples
that fit into our framework are models which are 
constructed via the perturbative BRST/BV formalism for AQFT, see 
e.g.\ \cite{Hollands:2007zg,Fredenhagen:2011an,Fredenhagen:2011mq,Tehrani:2018kbs}.
We will assume throughout the whole section that $k\supseteq \mathbb{Q}$ 
is a commutative unital ring that includes the ring of rational numbers 
as a subring. (The physically relevant examples are complex numbers
$k=\bbC$ and formal power series $k=\bbC[[\hbar]]$.)
This will considerably simplify our model categorical
considerations and arguments. For details on the material
presented below we refer to  \cite{Benini:2018oeh}.


\subsection{Homotopy theory of algebras over dg-operads}
Colored dg-operads are similar to the $\Set$-valued colored operads from
Section \ref{subsec:AQFToperad}, however with the difference that
they have chain complexes of $n$-ary operations. In more detail,
a {\em colored dg-operad} $\O\in \Op(\Ch(k))$ is described by the following data:
\begin{enumerate}[i)]
\item an underlying set of colors;
\item for each tuple $(\und{c},t) = ((c_1,\dots,c_n),t)$ of colors,
a chain complex $\O\big(\substack{t \\ \und{c}}\big)\in\Ch(k)$
of operations from $\und{c}$ to $t$;
\item composition $\Ch(k)$-morphisms\newline
$\gamma : \O\big(\substack{t \\ \und{c}}\big) \otimes\bigotimes_{i=1}^n\O\big(\substack{c_i \\ \und{b}_i}\big)
\to \O\big(\substack{t \\ (\und{b}_1,\dots,\und{b}_n)}\big)$;
\item unit $\Ch(k)$-morphisms $\oone : k\to \O\big(\substack{t \\ t}\big)$;
\item permutation action $\Ch(k)$-morphisms $\O(\sigma) : \O\big(\substack{t \\ \und{c}}\big)\to
\O\big(\substack{t \\ \und{c}\sigma}\big)$.
\end{enumerate}
This data has to satisfy the usual associativity, unitality and equivariance
conditions, see e.g.\ \cite{yau2016colored}. Given any $\Set$-valued colored operad
$\P\in\Op(\Set)$, one can define a colored dg-operad
$\P\otimes k\in\Op(\Ch(k))$ by tensoring 
\begin{flalign}
\P\big(\substack{t \\ \und{c}}\big)\otimes k\, :=\, \coprod\limits_{p\in \P\big(\substack{t \\ \und{c}}\big)} k
\end{flalign}
each set of operations with the monoidal unit $k\in\Ch(k)$.

To every colored dg-operad $\O\in \Op(\Ch(k))$ one can assign its category of algebras
$\Alg_\O(\Ch(k))$ with values in the symmetric monoidal model category $\Ch(k)$
of chain complexes. Concretely, an {\em $\O$-algebra} $A\in\Alg_\O(\Ch(k))$ 
is a collection of chain complexes $A_c\in\Ch(k)$, for all colors $c$, together with $\Ch(k)$-morphisms
\begin{flalign}
A \,:\, \O\big(\substack{t \\ \und{c}} \big) \otimes \bigotimes\limits_{i=1}^n A_{c_i} \longrightarrow A_t
\end{flalign}
that encode the actions of the chain complexes of operations. Of course, various compatibility
conditions with the operad structure on $\O$ must be fulfilled, see e.g.\ \cite{yau2016colored}.

Recalling that $\Ch(k)$ is a (symmetric monoidal) model category, with
weak equivalences the quasi-isomor-\linebreak phisms, it is natural to ask whether
$\Alg_\O(\Ch(k))$ is a model category too. In general, this turns out to be
a complicated question and there is a large amount of literature
on model structures for operad algebras in model categories, see e.g.\
\cite{Hinich:9702015,Berger:0206094,berger2007resolution,Spitzweck:2001aa,Cisinski:1109.1004,Hinich:1311.4130,Pavlov:1410.5675}.
The case of relevance to us has been understood by Hinich \cite{Hinich:9702015,Hinich:1311.4130}, who has proven
the following result.
\begin{theo}\label{theo:Oalgebramodelstructure}
Let $\O\in\Op(\Ch(k))$ be any colored dg-operad. 
Define a morphism $\kappa : A\to B$ in $\Alg_{\O}(\Ch(k))$ to be
\begin{enumerate}[i)]
\item a weak equivalence if each component $\kappa_c : A_c\to B_c$ is a weak
equivalence in $\Ch(k)$, i.e.\ a quasi-isomorphism;
\item a fibration if each component $\kappa_c : A_c\to B_c$ is a fibration
in $\Ch(k)$, i.e.\ degree-wise surjective;
\item a cofibration if it has the left lifting property with respect
to all acyclic fibrations.
\end{enumerate}
If $k\supseteq \mathbb{Q}$, then these choices endow $\Alg_{\O}(\Ch(k))$ 
with the structure of a model category. In this model structure
every object $A\in \Alg_{\O}(\Ch(k))$ is fibrant.
\end{theo}

Let us now consider any $\Op(\Ch(k))$-morphism
$\phi : \O\to \P$ between two colored dg-operads.
There is an associated pullback functor
$\phi^\ast : \Alg_{\P}(\Ch(k))\to \Alg_\O(\Ch(k))$
between the categories of algebras, which admits a left
adjoint given by operadic left Kan extension, i.e.\ we obtain an adjunction
\begin{flalign}\label{eqn:phiOalgebraadjunction}
\xymatrix{
\phi_! \,:\, \Alg_{\O}(\Ch(k)) ~\ar@<0.5ex>[r] & \ar@<0.5ex>[l]~ \Alg_\P(\Ch(k)) \,:\, \phi^\ast
}\quad.
\end{flalign}
It is easy to see that this adjunction is compatible with the model structures
in Theorem \ref{theo:Oalgebramodelstructure}.
\begin{propo}\label{prop:OalgebraQuillenadjunction}
For every $\Op(\Ch(k))$-morphism $\phi : \O\to \P$,
the adjunction \eqref{eqn:phiOalgebraadjunction} is a Quillen adjunction.
Moreover, the right adjoint $\phi^\ast$ preserves weak equivalences.
\end{propo}

\begin{rem}\label{rem:derivedOpKanextension}
In general, the left adjoint functor $\phi_!$ in  \eqref{eqn:phiOalgebraadjunction} 
does not preserve weak equivalences
and has to be derived. Choosing any natural cofibrant replacement
$\big(Q : \Alg_{\O}(\Ch(k)) \to \Alg_{\O}(\Ch(k)), q :  Q\stackrel{\sim}{\to} \id\big)$,
we define as usual the left derived functor by
\begin{flalign}
\bbL \phi_! \,:=\, \phi_!\, Q\,:\, \Alg_{\O}(\Ch(k))\longrightarrow \Alg_{\P}(\Ch(k))\quad.
\end{flalign}
By construction, $\bbL \phi_! $ does preserve weak equivalences.
Because by Theorem \ref{theo:Oalgebramodelstructure} 
each object in $\Alg_\P(\Ch(k))$ is fibrant, we will always choose in what
follows the trivial fibrant replacement
$\big(R=\id : \Alg_{\P}(\Ch(k)) \to \Alg_{\P}(\Ch(k)), r=\id  : \id \stackrel{\sim}{\to} R \big)$,
which means that the right derived functor
\begin{flalign}
\bbR\phi^\ast \,:= \,\phi^\ast\, R \,=\, \phi^\ast \,:\,\Alg_{\P}(\Ch(k))\longrightarrow  \Alg_{\O}(\Ch(k))
\end{flalign}
agrees with the underived functor.
\end{rem}

Next, we shall briefly discuss the concept of $\Sigma$-cofi\-brant resolutions
of colored dg-operads and the associated concept of homotopy algebras.
In what follows we fix a non-empty set $\CCC\in\Set$ of colors
and consider the subcategory $\Op_{\CCC}(\Ch(k))\subseteq \Op(\Ch(k))$ of 
colored dg-operads with fixed set of colors $\CCC$ and morphisms acting as the identity on colors.
Under our hypothesis $k\supseteq \mathbb{Q}$,
Hinich's results \cite{Hinich:9702015,Hinich:1311.4130} imply that $\Op_{\CCC}(\Ch(k))$
is a model category with weak equivalences and fibrations defined
component-wise and cofibrations defined by the left lifting property.
Let us introduce the following standard terminology.
\begin{defi}\label{def:sigmacofibrant}
A $\CCC$-colored dg-operad $\O\in\Op_{\CCC}(\Ch(k))$
is called {\em $\Sigma$-cofibrant} if each component
$\O\big(\substack{t \\ \und{c}}\big)$ is a cofibrant object in
the projective model structure on the functor category $\Ch(k)^{\Sigma_{\und{c}}}$,
where $\Sigma_{\und{c}}\subseteq \Sigma_n$ is the stabilizer subgroup of
the tuple $\und{c} = (c_1,\dots,c_n)$ of colors.
\end{defi}

Cofibrant dg-operads, i.e.\ cofibrant objects in the model category 
$\Op_{\CCC}(\Ch(k))$, are in particular $\Sigma$-cofibrant, 
cf.\ \cite[Proposition 4.3]{Berger:0206094}.
However, the converse is not true since e.g.\ the commutative 
dg-operad $\mathsf{Com}$ is $\Sigma$-cofibrant but not cofibrant. 
The relevance of $\Sigma$-cofibrant dg-operads is that their categories
of algebras behave well with respect to weak equivalences.
More precisely, the relevant result \cite{Hinich:9702015,Hinich:1311.4130} is as follows.
\begin{theo}\label{theo:Sigmacofibrantwe}
Let $\phi : \O\to \P$ be a weak equivalence between $\Sigma$-cofibrant 
colored dg-operads $\O,\P\in\Op_\CCC(\Ch(k))$. Then the corresponding
Quillen adjunction \eqref{eqn:phiOalgebraadjunction}  is a Quillen equivalence.
\end{theo}

With this preparation we can now finally define the concept of 
homotopy algebras over colored dg-operads.
\begin{defi}\label{def:sigmsresolution} 
Let $\O\in\Op_\CCC(\Ch(k))$  be a colored dg-operad.
\begin{enumerate}[i)]
\item A {\em $\Sigma$-cofibrant resolution} of $\O$ is a $\Sigma$-cofibrant
colored dg-operad $\O_\infty\in\Op_{\CCC}(\Ch(k))$ together with
an acyclic fibration $w : \O_\infty \to\O$ in $\Op_\CCC(\Ch(k))$.

\item The model category of {\em homotopy $\O$-algebras}
is the model category $\Alg_{\O_\infty}(\Ch(k))$ of algebras
over a $\Sigma$-cofibrant resolution $w:\O_\infty\to \O$.
\end{enumerate}
\end{defi}

\begin{rem}
It is natural to ask whether the concept of homotopy $\O$-algebras depends
on the chosen resolution. Given two $\Sigma$-cofibrant resolutions 
$w:\O_\infty\to \O$ and $w^\prime : \O^\prime_\infty\to \O$,
and taking also a cofibrant replacement $q : Q\O\to \O$,
we obtain a commutative diagram
\begin{flalign}
\xymatrix@R=1em@C=2em{
\O_\infty \ar[r]^-{w} & \O & \ar[l]_-{w^\prime} \O_\infty^\prime\\
& \ar@{-->}[lu]^-{l}Q\O\ar[u]^-{q} \ar@{-->}[ru]_-{l^\prime}&
}
\end{flalign}
in $\Op_\CCC(\Ch(k))$. The dashed arrows exist due to the left lifting property, 
because $Q\O$ is by construction a cofibrant dg-operad and $w,w^\prime$ are acyclic fibrations.
By Proposition \ref{prop:OalgebraQuillenadjunction} we obtain a zig-zag
\begin{flalign}
\resizebox{.9\hsize}{!}{$\xymatrix@C=1em{
\Alg_{\O_\infty}(\Ch(k)) \ar@<-0.5ex>[r]_-{l^\ast} & \ar@<-0.5ex>[l]_-{l_!}~ \Alg_{Q\O}(\Ch(k))
\ar@<0.5ex>[r]^-{l^\prime_!} & \ar@<0.5ex>[l]^-{l^{\prime \ast}} \Alg_{\O^\prime_\infty}(\Ch(k))
}$}
\end{flalign}
of Quillen adjunctions. Because $\O_\infty$, $\O_{\infty}^\prime$ and also the cofibrant
dg-operad $Q\O$ are $\Sigma$-cofibrant, Theorem \ref{theo:Sigmacofibrantwe}
implies that this is a zig-zag of Quillen equivalences and hence
that the model categories $\Alg_{\O_\infty}(\Ch(k))$ and $\Alg_{\O_\infty^\prime}(\Ch(k))$
of homotopy $\O$-algebras are equivalent in this sense.
\end{rem}

\begin{ex}
The following are standard examples of homotopy algebras:
\begin{enumerate}[i)]
\item $A_\infty$-algebras are homotopy algebras over the associative operad
$\mathsf{As}$;
\item $E_\infty$-algebras are homotopy algebras over the commutative operad
$\mathsf{Com}$;
\item $L_\infty$-algebras are homotopy algebras over the Lie operad $\mathsf{Lie}$;
\item homotopy-coherent diagrams are homotopy algebras over the diagram
operad $\mathsf{Diag}_\CC$.
\end{enumerate}
\end{ex}


\subsection{\label{subsec:AQFTmodelcat}AQFT model categories and Quillen adjunctions}
Let $\ovr{\CC}= (\CC,\perp)$ be any orthogonal category
and $\O_{\ovr{\CC}}\in \Op(\Set)$ the corresponding
AQFT operad from Theorem \ref{theo:AQFToperad}.
The following result is fundamental for the developments 
throughout the whole section.
\begin{theo}\label{theo:AQFTmodelcal}
For every orthogonal category $\ovr{\CC}$, the category
$\QFT(\ovr{\CC})$ of $\Ch(k)$-valued AQFTs on $\ovr{\CC}$
(cf.\ Definition \ref{def:QFTcats}) is a model category
with respect to the following choices: A morphism
$\zeta : \AAA\to \BBB$ in $\QFT(\ovr{\CC})$ (i.e.\ a natural
transformation between the underlying functors $\CC\to\Alg_{\As}(\Ch(k))$) is
\begin{enumerate}[i)]
\item a weak equivalence if each component $\zeta_c : \AAA(c)\to\BBB(c)$ is a quasi-isomorphism;

\item a fibration if each component $\zeta_c : \AAA(c)\to \BBB(c)$ is degree-wise surjective;

\item a cofibration if it has the left lifting property with respect to all acyclic fibrations.
\end{enumerate}
\end{theo}
\begin{proof}
This is a direct consequence of Theorem \ref{theo:Oalgebramodelstructure}
and the chain of isomorphism $\QFT(\ovr{\CC})\cong \Alg_{\O_{\ovr{\CC}}}(\Ch(k))\cong$\linebreak $\Alg_{\O_{\ovr{\CC}}\otimes k}(\Ch(k))$.
The first step is Theorem \ref{theo:AQFTalgebraoperad}
and for the second step one easily checks that the category 
$\Alg_{\O_{\ovr{\CC}}}(\Ch(k))$ of $\Ch(k)$-valued algebras
over the $\Set$-valued AQFT operad $\O_{\ovr{\CC}}\in \Op(\Set)$
is (isomorphic to) the category $\Alg_{\O_{\ovr{\CC}}\otimes k}(\Ch(k))$
of algebras over the corresponding AQFT dg-operad 
$\O_{\ovr{\CC}}\otimes k \in\Op(\Ch(k))$. 
\end{proof}

\begin{rem}
Our rigorous concept of weak equivalences for $\Ch(k)$-valued AQFTs
established in Theorem \ref{theo:AQFTmodelcal} agrees with the less
formal notions used in concrete applications
of the BRST/BV formalism, cf.\ \cite{Hollands:2007zg,Fredenhagen:2011an,Fredenhagen:2011mq}.
In particular, the usual technique of `adding auxiliary fields without changing
homologies' can be understood rigorously from our model categorical
perspective.
\end{rem}

In the context of $\Ch(k)$-valued AQFTs, the universal constructions
from Section \ref{subsec:universal} have to be derived in order to
be consistent with the concept of weak equivalences introduced in
Theorem \ref{theo:AQFTmodelcal}. Recall  from
Theorem \ref{theo:algadjunction} that every
orthogonal functor $F: \ovr{\CC}\to\ovr{\DD}$ defines
an adjunction
\begin{flalign}\label{eqn:adjunctionCH}
\xymatrix{
F_! \,:\, \QFT(\ovr{\CC}) ~\ar@<0.5ex>[r] & \ar@<0.5ex>[l]~ \QFT(\ovr{\DD}) \,:\, F^\ast
}
\end{flalign}
between the corresponding categories of $\Ch(k)$-valued AQFTs.
As a direct consequence of Proposition \ref{prop:OalgebraQuillenadjunction},
we obtain the following result.
\begin{propo}\label{propo:AQFTQuillenAdjunction}
For every orthogonal functor $F : \ovr{\CC}\to\ovr{\DD}$,
the adjunction \eqref{eqn:adjunctionCH} is a Quillen adjunction.
Moreover, the right adjoint $F^\ast$ preserves weak equivalences.
\end{propo}

As explained in Remark \ref{rem:derivedOpKanextension},
the left adjoint functor $F_!$ has to be derived, e.g.\ by choosing
a natural cofibrant replacement. This has consequences
for the examples of universal constructions
discussed in Section \ref{subsec:universal}.
\begin{ex}\label{ex:homotopyjlocal}
Let $j:\ovr{\CC}\to\ovr{\DD}$ be a full orthogonal subcategory embedding
and consider as in \eqref{eqn:jadjunction} the extension/restriction
adjunction $j_! \dashv j^\ast$ for $\Ch(k)$-valued AQFTs.
By Proposition \ref{propo:AQFTQuillenAdjunction}, this is a Quillen
adjunction and hence we can construct the left derived extension functor
$\bbL j_! := j_!\,Q : \QFT(\ovr{\CC})\to\QFT(\ovr{\DD})$. 
We would like to emphasize again that it is the derived functor
$\bbL j_!$ that defines a meaningful local-to-global 
extension for $\Ch(k)$-valued AQFTs
and {\em not} the underived functor $j_!$, because the latter
in general does not preserve weak equivalences. (See 
\cite[Appendix A]{Benini:2018oeh} for concrete examples.)
This in particular means that our definition of descent
via $j$-locality (cf.\ Definition \ref{def:jlocal}) has to be adapted 
in order to be homotopically meaningful. Following \cite{Bruinsma:2018knq},
we say that $\AAA \in \QFT(\ovr{\DD})$ is {\em homotopy $j$-local}
if the corresponding component
\begin{flalign}
\xymatrix{
\bbL j_! j^\ast \AAA = j_! Q j^\ast\AAA \ar[r]^-{j_! q_{j^\ast\AAA}} & j_! j^\ast \AAA \ar[r]^-{\epsilon_{\AAA}} & \AAA
}
\end{flalign}
of the derived counit is a weak equivalence. It is easy to prove
that the derived extension $\bbL j_! \BBB $ of every $\BBB\in\QFT(\ovr{\CC})$
is homotopy $j$-local. Toy-models of homotopy $j$-local
AQFTs that are inspired by gauge theory are presented in 
Section \ref{subsec:derivedlocaltoglobal} below.
\end{ex}

\begin{ex}
Let $L : \ovr{\CC}\to\ovr{\CC[W^{-1}]}$ be an orthogonal localization
and consider as in \eqref{eqn:Ladjunction} the time-slicification
adjunction $L_!\dashv L^\ast$ for $\Ch(k)$-valued AQFTs.
By Proposition \ref{propo:AQFTQuillenAdjunction}, this is a Quillen
adjunction and hence we can construct the left derived time-slicification functor
$\bbL L_! := L_!\,Q : \QFT(\ovr{\CC})\to\QFT(\ovr{\CC[W^{-1}]})$. 
Our concept of $W$-constancy from Corollary \ref{cor:Wconstancy}
has to be adapted in order to be homotopically meaningful.
Following \cite{Bruinsma:2018knq}, we say that
$\AAA\in \QFT(\ovr{\CC})$ is {\em homotopy $W$-constant}
if the corresponding component
\begin{flalign}
\xymatrix{
Q\AAA \ar[r]^-{\eta_{Q\AAA}} & L^\ast L_! Q\AAA = L^\ast \bbL L_!\AAA
}
\end{flalign}
of the derived unit is a weak equivalence. Note that 
homotopy $W$-constancy can be interpreted as
a homotopy theoretic generalization of the time-slice axiom.
We expect that this will be useful for formalizing
the weaker concept of time-slice axiom appearing
in derived geometry, see the end of Section \ref{subsec:derivedgeometry}.
\end{ex}

\begin{problem}
The formal properties of the derived unit and counit of the time-slicification
adjunction $L_! \dashv L^\ast$ are harder to understand
than the ones for the extension/restriction
adjunction $j_! \dashv j^\ast$. In particular,
even though the {\em underived}
counit $\epsilon$ is a natural isomorphism by Proposition \ref{propo:Lreflective},
it is unclear if the {\em derived} counit
\begin{flalign}
\xymatrix{
\bbL L_! L^\ast = L_! Q L^\ast \ar[r]^-{L_! q L^\ast} & L_! L^\ast \ar[r]^-{\epsilon} & \id
}
\end{flalign}
is a natural weak equivalence. As a consequence, it is currently unclear 
to us if theories of the form $L^\ast \BBB$, for $\BBB\in\QFT(\ovr{\CC[W^{-1}]})$,
are homotopy $W$-constant.
\end{problem}


\subsection{Homotopy-coherent AQFTs}
The aim of this section is to study homotopy algebras
over the AQFT dg-operad $\O_{\ovr{\CC}}\otimes k\in\Op(\Ch(k))$,
which we shall also call {\em homotopy AQFTs}.
The following fundamental theorem is proven in \cite{Benini:2018oeh}.
Recall from Definitions \ref{def:sigmacofibrant} and \ref{def:sigmsresolution} 
the concepts of $\Sigma$-cofibrant dg-operads and $\Sigma$-cofibrant resolutions.
\begin{theo}\label{theo:strictification}
Let us assume as before that $k\supseteq \mathbb{Q}$ includes the rationals.
For every orthogonal category $\ovr{\CC}$, the AQFT dg-operad
$\O_{\ovr{\CC}}\otimes k\in\Op(\Ch(k))$ is $\Sigma$-cofibrant.
As a consequence of Theorem \ref{theo:Sigmacofibrantwe},
every $\Sigma$-cofibrant resolution $w : \O_\infty \to \O_{\ovr{\CC}}\otimes k$
induces a Quillen equivalence
\begin{flalign}
\xymatrix{
w_! \,:\, \Alg_{\O_\infty}(\Ch(k)) ~\ar@<0.5ex>[r] & \ar@<0.5ex>[l]~ \QFT(\ovr{\CC}) \,:\, w^\ast
}\quad.
\end{flalign}
\end{theo}

\begin{rem}
This result can be interpreted as a {\em strictification theorem} for homotopy AQFTs.
Indeed, given any homotopy AQFT $\AAA_\infty\in \Alg_{\O_\infty}(\Ch(k))$,
a cofibrant replacement and the derived unit of $w_! \dashv w^\ast$ defines
a zig-zag
\begin{flalign}
\xymatrix{
\AAA_\infty & \ar[l]_-{q_{\AAA_{\infty}}} Q\AAA_\infty \ar[r]^-{\eta_{Q\AAA_\infty}} & w^\ast \bbL w_! \AAA_\infty
}
\end{flalign}
of weak equivalences between $\AAA_\infty$ and the strict AQFT
$w^\ast \bbL w_! \AAA_\infty$.
\end{rem}

So does this mean that homotopy AQFTs are not interesting and important 
at all? The answer to this question is clearly no, because certain interesting
constructions naturally define non-strict homotopy AQFTs. For instance,
let us recall from Example \ref{ex:evidencehocostructures} that
taking derived smooth normalized cochain algebras on a functor
$\mathfrak{F} : \CC^\op\to\mathbf{St}_\infty\subseteq \HH_\infty$
that assigns $\infty$-stacks of gauge fields to space-times
defines a functor $\AAA := \bbL N^{\infty\ast}(\mathfrak{F}(-),k):
\CC\to\Alg_{\E_\infty}(\Ch(k))$ with values in $E_\infty$-algebras.
We shall show in Example \ref{ex:classicalAQFTs} 
below that this can be interpreted as a non-strict 
homotopy AQFT. A further class of examples
is given in Section \ref{subsec:fiberedcohomology}.
Note that even though each of these non-strict homotopy AQFTs
can be strictified by Theorem \ref{theo:strictification}, such strictifications
are hard to describe explicitly and thus it is often useful in practice 
to work directly with the weaker model.

There exist of course many different $\Sigma$-cofibrant resolutions
of the AQFT dg-operad $\O_{\ovr{\CC}}\otimes k \in\Op(\Ch(k))$,
which describe homotopy AQFTs whose algebraic structures
(functoriality, associativity, $\perp$-commutativity, etc.) 
are weakened in a homotopy-coherent sense by a certain extent.
The strictest possible resolution is given by the identity
$\id: \O_{\ovr{\CC}}\otimes k\to \O_{\ovr{\CC}}\otimes k$
and homotopy AQFTs with respect to this resolution are precisely strict AQFTs.
A very weak resolution, called the Boardman-Vogt resolution,
has been studied for our AQFT operads by Yau in \cite{Yau:2018dnm}.
The resulting homotopy AQFTs are, roughly speaking, 
homotopy-coherent diagrams of $A_\infty$-algebra
that satisfy a homotopy-coherent $\perp$-commutativity property.
Motivated by our examples from Example \ref{ex:evidencehocostructures} 
and Section \ref{subsec:fiberedcohomology}, we shall study below a 
particular $\Sigma$-cofibrant resolution $w : \O_{\ovr{\CC}}\otimes \E_\infty\to
\O_{\ovr{\CC}}\otimes k$ of the AQFT dg-operad that is obtained by a component-wise
tensoring with the Barratt-Eccles $E_\infty$-operad $\E_\infty\in \Op(\Ch(k))$.
This describes homotopy AQFTs that are strictly functorial, however
with a homotopy-coherent $\perp$-commutativity property.

Without going into any details, let us recall from \cite{Berger:0109158}
that the Barratt-Eccles dg-operad $\E_\infty\in\Op(\Ch(k))$
is a $\Sigma$-cofibrant resolution $w : \E_\infty\to \mathsf{Com}$
of the commutative dg-operad. The usual $\Op(\Ch(k))$-morphism
$\As\to\mathsf{Com}$ from the associative to the commutative dg-operad
factors through $\E_\infty$, i.e.\ we have a chain of operad maps
\begin{flalign}
\xymatrix{
\As \ar[r]^-{i} & \E_\infty  \ar[r]^-{w} & \mathsf{Com}
}\quad.
\end{flalign}
The induced Quillen adjunctions imply that each commutative
dg-algebra $C\in \Alg_{\mathsf{Com}}(\Ch(k))$ can be interpreted
as a strictly commutative $E_\infty$-algebra\linebreak $w^\ast C\in\Alg_{\E_\infty}(\Ch(k))$. 
Moreover, each $E_\infty$-algebra
$A\in\Alg_{\E_\infty}(\Ch(k))$ has an underlying 
dg-algebra $i^\ast A\in\Alg_{\As}(\Ch(k))$,
which is in general noncommutative unless the $E_\infty$-algebra
is strictly commutative.

For every orthogonal category $\ovr{\CC}$, we define a
colored dg-operad $\O_{\ovr{\CC}}\otimes \E_\infty \in \Op(\Ch(k))$
by the component-wise tensoring
\begin{flalign}
\big(\O_{\ovr{\CC}}\otimes \E_\infty\big)\big(\substack{t \\ \und{c}}\big)\,:=\,
\O_{\ovr{\CC}}\big(\substack{t \\ \und{c}}\big)\otimes \E_{\infty}(n)\in\Ch(k)
\end{flalign}
of the AQFT operad with the Barratt-Eccles $E_\infty$-operad, 
where $n$ is the length of the tuple of colors $\und{c}$.
The following result was proven in \cite{Benini:2018oeh}.
\begin{theo}\label{theo:AQFTresolution}
Let us assume as before that $k\supseteq \mathbb{Q}$ includes the rationals.
For every orthogonal category $\ovr{\CC}$, the $\Op(\Ch(k))$-morphism
\begin{flalign}
w_{\ovr{\CC}}\, := \, \id\otimes w \,:\,  \O_{\ovr{\CC}}\otimes \E_\infty 
\longrightarrow \O_{\ovr{\CC}}\otimes \mathsf{Com}
\cong \O_{\ovr{\CC}}\otimes k
\end{flalign}
defines a $\Sigma$-cofibrant resolution of the AQFT dg-operad.
These resolutions are natural in the orthogonal category $\ovr{\CC}\in\OCat$.
\end{theo}

\begin{defi}\label{def:EinftyhAQFT}
Let $\ovr{\CC}$ be an orthogonal category. We denote by
\begin{flalign}
\QFT_\infty(\ovr{\CC})\,:=\, \Alg_{\O_{\ovr{\CC}}\otimes\E_\infty}(\Ch(k))
\end{flalign}
the {\em model category of homotopy AQFTs} on $\ovr{\CC}$ that correspond
to the $\Sigma$-cofibrant resolution from Theorem \ref{theo:AQFTresolution}.
\end{defi}

\begin{rem}
Our universal constructions for strict $\Ch(k)$-valued AQFTs from 
Section \ref{subsec:AQFTmodelcat} immediately generalize to the case
of homotopy AQFTs. In particular, for every orthogonal
functor $F : \ovr{\CC}\to\ovr{\DD}$ one obtains a Quillen adjunction
\begin{flalign}
\xymatrix{
F_! \,:\, \QFT_\infty(\ovr{\CC}) ~\ar@<0.5ex>[r] & \ar@<0.5ex>[l]~ \QFT_\infty(\ovr{\DD}) \,:\, F^\ast
}
\end{flalign}
between the corresponding model categories of homotopy AQFTs. Interesting
examples are again extension/restriction adjunctions induced by
full orthogonal subcategory embeddings or 
time-slicification adjunctions induced by orthogonal localizations.
\end{rem}

\begin{ex}\label{ex:classicalAQFTs}
Given any small category $\CC$, we can choose the maximal
orthogonality relation 
\begin{flalign}
\perp^{\mathrm{max}}  \,:=\,\mathrm{Mor}\CC \,{}_{\mathrm{t}}{\times}_{\mathrm{t}}\, \mathrm{Mor}\CC
\end{flalign}
and define an orthogonal category
$\ovr{\CC}^{\mathrm{max}} := (\CC,\perp^{\mathrm{max}})$.
One easily checks that the category 
\begin{flalign}
\QFT(\ovr{\CC}^{\mathrm{max}}) \,\cong\, \Alg_{\mathsf{Com}}(\Ch(k))^\CC
\end{flalign}
of strict AQFTs on $\ovr{\CC}^{\mathrm{max}}$ is the category
of functors from $\CC$ to commutative dg-algebras
and that the category
\begin{flalign}
\QFT_\infty(\ovr{\CC}^{\mathrm{max}}) \,\cong\, \Alg_{\E_\infty}(\Ch(k))^\CC
\end{flalign}
of homotopy AQFTs (in the sense of Definition \ref{def:EinftyhAQFT})
on $\ovr{\CC}^{\mathrm{max}}$ is the category
of functors from $\CC$ to $E_\infty$-algebras.
In particular, the 
derived smooth normalized cochain algebras
on diagrams of $\infty$-stacks from Example \ref{ex:evidencehocostructures}
define examples of such homotopy AQFTs.
\end{ex}


\subsection{\label{subsec:derivedlocaltoglobal}Derived local-to-global constructions}
In this section we shall present concrete results on derived local-to-global
extensions in a simplified setting. In particular, we shall show that certain
simplified toy-models for topological AQFTs satisfy a homotopy $j$-locality
property in the sense of Example \ref{ex:homotopyjlocal}. Physically, these results
should be interpreted as a homotopical descent condition for such AQFTs.
For the technical details we refer to \cite{Benini:2018oeh}.

Let $\Man$ be the category of oriented $m$-dimensional manifolds
of finite type with morphisms given by orientation preserving open embeddings.
We endow $\Man$ with the maximal orthogonality relation from Example
\ref{ex:classicalAQFTs}, which defines an orthogonal category
$\ovr{\Man}^{\mathrm{max}} := (\Man,\perp^{\mathrm{max}})$.
Let further $\ovr{\Disk}^{\mathrm{max}}\subseteq \ovr{\Man}^{\mathrm{max}}$
be the full orthogonal subcategory of all manifolds diffeomorphic
to $\bbR^m$ and denote the corresponding 
full orthogonal subcategory embedding by $j: \ovr{\Disk}^{\mathrm{max}}
\to \ovr{\Man}^{\mathrm{max}}$. We are interested
in describing the left derived functor $\bbL j_!$ 
of the associated Quillen adjunction
\begin{flalign}
\resizebox{.86\hsize}{!}{$
\xymatrix@C=1.5em{
j_! \,:\, \QFT_\infty(\ovr{\Disk}^{\mathrm{max}}) \ar@<0.5ex>[r] & \ar@<0.5ex>[l]
\QFT_\infty(\ovr{\Man}^{\mathrm{max}}) \,:\, j^\ast
}~~.$}
\end{flalign}
The following technical theorem is the key ingredient for our computations.
Its proof uses Lurie's Seifert-van Kampen theorem \cite[Appendix A.3.1]{Lurie:book}
and is presented in \cite{Benini:2018oeh}.
\begin{theo}\label{theo:SvK}
Suppose that $\AAA \in \QFT_\infty(\ovr{\Disk}^{\mathrm{max}})$ is\linebreak weakly equivalent
to a constant functor $\Disk\to$\linebreak $\Alg_{\E_\infty}(\Ch(k))$ whose value we denote by
$A\in\Alg_{\E_\infty}(\Ch(k))$. Then the derived extension
$\bbL j_! \AAA \in \QFT_\infty(\ovr{\Man}^{\mathrm{max}})$ may be
computed object-wise for each $M\in \Man$ by
\begin{flalign}
\big(\bbL j_! \AAA \big)(M)\,=\, \mathrm{Sing}(M)\stackrel{\bbL}{\otimes} A \,\in\Alg_{\E_\infty}(\Ch(k))\quad,
\end{flalign}
where $\mathrm{Sing}(M)\in\sSet$ is the simplicial set of singular simplices
in $M$ and $\stackrel{\bbL}{\otimes}$ is the derived $\sSet$-tensoring
for $E_\infty$-algebras, cf.\ \cite{Fresse:1503.08489,Ginot:1011.6483}.
\end{theo}

\begin{rem}
In \cite{Ginot:1011.6483} the $E_\infty$-algebra
$\mathrm{Sing}(M)\stackrel{\bbL}{\otimes} A$ is also referred
to as the {\em derived higher Hochschild chains} on  $\mathrm{Sing}(M)$
with coefficients in $A$.
\end{rem}

\begin{ex}
Our first example is inspired by Dijkgraaf-Witten theory. Let us consider a gauge theory
whose fields on $M$ are described by the groupoid $\mathfrak{Bun}_G(M)\in\Grpd$
of principal $G$-bundles on $M$ for a finite nilpotent group $G$. By \cite[Lemma 2.8]{Schweigert:2018rmn},
the nerve $B \mathfrak{Bun}_G(M)\in\sSet$ is weakly equivalent
to the simplicial mapping space $BG^{\mathrm{Sing}(M)}\in\sSet$, 
where $BG\in\sSet$ is the nerve of the groupoid $\ast//G$. We define
a homotopy AQFT $\AAA\in \QFT_\infty(\ovr{\Man}^{\mathrm{max}})$
by forming on each $M\in\Man$ the normalized cochain algebra
\begin{flalign}
\AAA(M) \,:=\, N^\ast\big(BG^{\mathrm{Sing}(M)},k\big)\in\Alg_{\E_\infty}(\Ch(k))\quad.
\end{flalign}
The restriction $j^\ast \AAA\in\QFT_\infty(\ovr{\Disk}^{\mathrm{max}})$  to disks
is weakly equivalent to a constant functor with value $N^\ast(BG,k)$ because
$\mathrm{Sing}(U)\to \ast$ is a weak equivalence for every $U\in\Disk$. Hence,
we can apply Theorem \ref{theo:SvK} and compute
\begin{flalign}
\nonumber \big(\bbL j_!j^\ast \AAA\big)(M) &\simeq\mathrm{Sing}(M)\stackrel{\bbL}{\otimes} N^\ast(BG,k)
\\ &\simeq N^\ast\big(BG^{\mathrm{Sing}(M)},k\big) = \AAA(M)\quad,
\end{flalign}
for all $M\in\Man$. The second step follows from
\cite[Proposition 5.3]{Ayala:1206.5522} and it uses that $G$ is a finite nilpotent group.
Summing up, we have seen that the present toy-model of a homotopy 
AQFT $\AAA\in \QFT_\infty(\ovr{\Man}^{\mathrm{max}})$ on $\ovr{\Man}^{\mathrm{max}}$ 
is weakly equivalent to the derived extension $\bbL j_!$ of its restriction to disks.
As a consequence, it is also homotopy $j$-local in the sense of Example \ref{ex:homotopyjlocal},
which one should interpret as a homotopical descent condition.
\end{ex}

\begin{ex}
Our second example is inspired by linear 
Chern-Simons theory with structure group $\bbR$, 
i.e.\ flat principal $\bbR$-bundles with connections on $2$-dimensional surfaces.
In what follows $\Man$ will denote the category
of $2$-dimensional oriented manifolds and we take $k=\bbR$. From the perspective
of derived geometry of linear gauge fields (cf.\ Section \ref{subsec:derivedgeometry}),
the linear classical observables for this theory on $M\in\Man$
are described by the $(-1)$-shifted compactly supported de Rham complex
\begin{flalign}
\Omega^\bullet_c(M)[-1]\,:= \, 
\Big(\!\!\!\!
\xymatrix@C=1em{
\stackrel{(-1)}{\Omega^2_c(M)} & \ar[l]_-{\dd} \stackrel{(0)}{\Omega^1_c(M)} & \ar[l]_-{\dd} 
\stackrel{(1)}{\Omega^0_c(M)}
}\!\!\!\!
\Big)\quad.
\end{flalign}
We define a homotopy AQFT $\AAA\in$\linebreak $\QFT_\infty(\ovr{\Man}^{\mathrm{max}})$
by forming on each $M\in\Man$ the free $E_\infty$-algebra
\begin{flalign}
\AAA(M)\,:=\, \mathbb{E}_\infty\big(\Omega^\bullet_c(M)[-1]\big)\,\in\Alg_{\E_\infty}(\Ch(k))
\end{flalign}
over this complex. One should interpret this as classical polynomial observables for linear Chern-Simons theory.
The restriction $j^\ast \AAA \in  \QFT_\infty(\ovr{\Disk}^{\mathrm{max}})$ to disks
is weakly equivalent to a constant functor with value $\mathbb{E}_\infty(\bbR[1])$ because the integration
map $\int_U : \Omega^\bullet_c(U)[-1]\to \bbR[1]$ is a weak equivalence for every $U\in\Disk$.
Hence, we can apply Theorem \ref{theo:SvK} and compute
\begin{flalign}
\nonumber
\big(\bbL j_!j^\ast \AAA\big)(M)&\simeq \mathrm{Sing}(M)\stackrel{\bbL}{\otimes} \mathbb{E}_\infty(\bbR[1])\\ 
\nonumber &\simeq \mathrm{Sing}(M)\otimes \mathbb{E}_\infty(\bbR[1]) \\
&\simeq \mathbb{E}_\infty\big(N_\ast(\mathrm{Sing}(M),\bbR)\otimes \bbR[1]\big)\quad,
\end{flalign}
for all $M\in\Man$.  In the second step we used that for free $E_\infty$-algebras
the derived $\sSet$-tensoring is weakly equivalent to the underived one. The third step
is a direct computation using the explicit formula for the latter from \cite{Fresse:1503.08489}.
One concludes that $\AAA\simeq \bbL j_! j^\ast \AAA$ are weakly equivalent
because the $1$-shifted $\bbR$-valued normalized chains $N_\ast(\mathrm{Sing}(M),\bbR)\otimes \bbR[1]$
are weakly equivalent to $\Omega^\bullet_c(M)[-1]$ as a consequence of de Rham's theorem.
Hence, this toy-model also satisfies the homotopy $j$-locality condition from 
Example \ref{ex:homotopyjlocal}.
\end{ex}

\begin{problem}
The examples considered above are only toy-models for 
the kind of homotopy AQFTs that we are eventually interested in.
This is because they are 1.)~`too topological' in the sense of being 
weakly equivalent to a constant diagram and 2.)~`not quantum' in
the sense that they assign only homotopy-coherently commutative observable 
algebras. It is an open problem to evaluate 
the derived extension functor $\bbL j_!$ and test homotopy $j$-locality 
for more realistic full orthogonal subcategory embeddings, e.g.\ 
$j: \ovr{\Locc}\to\ovr{\Loc}$ from Example \ref{ex:LocLocc}, 
which is crucial for Lorentzian AQFTs.
\end{problem}


\subsection{\label{subsec:fiberedcohomology}Examples from homotopy invariants}
We present another class of examples of homotopy AQFTs in the sense of
Definition \ref{def:EinftyhAQFT}. Let $\ovr{\CC}$ be an orthogonal category
and $\pi : \DD\to\CC$ a category fibered in groupoids over its underlying category $\CC$.
Endowing $\DD$ with the pullback orthogonality relation, we obtain
an orthogonal functor $\pi : \ovr{\DD}\to\ovr{\CC}$ that we call an
{\em orthogonal category fibered in groupoids}. The basic idea
behind our construction below is as follows: Given
any strict $\Ch(k)$-valued AQFT $\AAA \in\QFT(\ovr{\DD})$ on the total 
category $\ovr{\DD}$, we would like to construct an AQFT $\AAA_\pi$ 
on the base category $\ovr{\CC}$ by forming {\em homotopy invariants} along
the groupoid fibers $\pi^{-1}(c)\in\Grpd$, for all $c\in\CC$. We shall formalize
this construction and show that it naturally leads to a homotopy AQFT 
$\AAA_\pi\in\QFT_\infty(\ovr{\CC})$ on $\ovr{\CC}$.
\begin{rem}
The physical motivation behind this construction is as follows:
As usual, $\ovr{\CC}$ is interpreted as a category of space-times.
The total category $\ovr{\DD}$ of the orthogonal category fibered in groupoids
$\pi:\ovr{\DD}\to\ovr{\CC}$ should be interpreted as a category of space-times
with additional geometric structures (which we call {\em background fields}), 
e.g.\ spin structures, bundles and connections.
The functor $\pi$ forgets this extra structure and hence its fibers
$\pi^{-1}(c)\in\Grpd$ are the groupoids of background fields
on the space-time $c\in\CC$. Note that the morphisms in these groupoids
are interpreted as gauge transformations of the background fields, cf.\ Section \ref{subsec:groupoids}.
Our construction assigns to an AQFT $\AAA \in\QFT(\ovr{\DD})$ on 
space-times with  background fields a homotopy AQFT $\AAA_\pi\in\QFT_\infty(\ovr{\CC})$
on plain space-times. This is achieved by assigning
to each $c\in\CC$ the  homotopy invariants of $\AAA$ along the 
action of the  groupoid $\pi^{-1}(c)$ of background fields over $c$.
We refer to \cite{Benini:2016jfs} for concrete examples and
more details on the physical interpretation.
\end{rem}

Without loss of generality, we can focus on the case where
our category fibered in groupoids is given by the Grothendieck
construction $\pi : \CC_F \to \CC$ of a presheaf
of groupoids $F : \CC^\op\to \Grpd$ on $\CC$.
This follows from the strictification theorems in \cite{Hollander:0110247}.
Forming homotopy invariants along the groupoid fibers can be described
by a homotopy right Kan extension $\mathrm{hoRan}_\pi : \Ch(k)^{\CC_F} \to\Ch(k)^\CC$
along $\pi : \CC_F \to\CC$ of the underlying $\Ch(k)$-valued functor 
of an AQFT $\AAA\in\QFT(\ovr{\CC_F})$. By \cite{Benini:2018oeh,Benini:2016jfs}, 
we have the following explicit model.
\begin{propo}\label{prop:hoRan}
Let $F : \CC^\op\to\Grpd$ be a presheaf of groupoids and $\AAA : \CC_F\to\Ch(k)$
a chain complex valued functor on the corresponding Grothendieck construction.
Then the homotopy right Kan extension $\AAA_\pi := \mathrm{hoRan}_\pi \AAA : \CC\to \Ch(k)$
along the projection functor $\pi:\CC_F\to\CC$ can be computed object-wise 
by the end
\begin{flalign}
\AAA_\pi(c)\,=\,\int_{x\in F(c)} \big[N_\ast \big(B(F(c)\downarrow x),k\big),\AAA(c,x)\big]\quad,
\end{flalign}
for all $c\in\CC$. Here $N_\ast(-,k)$ denotes the normalized chain functor, $[-,-]$ the internal hom
in $\Ch(k)$ and $B(F(c)\downarrow x)\in\sSet$ the nerve of the over-category $F(c)\downarrow x$.
\end{propo}

\begin{rem}
Note that $\AAA_\pi(c)$ can also be understood as the homotopy limit
\begin{flalign}
\AAA_\pi(c)\,\simeq\, \holim\Big(\AAA\big\vert_{\pi^{-1}(c)} : F(c)\to \Ch(k)\Big)
\end{flalign}
of the restriction of $\AAA$ to the groupoid fiber $\pi^{-1}(c)\simeq F(c)$.
This is important for our interpretation of $\AAA_\pi = \mathrm{hoRan}_\pi\AAA$ 
as forming fiber-wise homotopy invariants.
\end{rem}

The main result of this section is that 
the collection of chain complexes $\AAA_\pi(c)$, for $c \in \CC$, 
obtained with the construction above carries the structure 
of a homotopy AQFT in a canonical way. The key ingredient for the proof
is the result in \cite{Berger:0109158} that the normalized chain complex
$N_\ast(S,k)$ of a simplicial set $S$ carries a canonical $\E_\infty$-{\em co}action. 
This fact, combined with the original 
$\O_{\ovr{\CC_F}}$-algebra structure on $\AAA\in\QFT(\ovr{\CC_F})$, 
leads to the next theorem. 
We refer to \cite{Benini:2018oeh} for a detailed proof.
\begin{theo}\label{theo:homotopyinvariants}
Let $\ovr{\CC}$ be an orthogonal category and
$F : \CC^\op\to\Grpd$ a presheaf of groupoids.
Consider the orthogonal category fibered in groupoids
$\pi: \ovr{\CC_F}\to\ovr{\CC}$ that is obtained
by the Grothendieck construction of $F$.
For every strict $\Ch(k)$-valued AQFT $\AAA\in\QFT(\ovr{\CC_F})$
on the total category, the family of chain complexes
$\AAA_\pi(c)\in\Ch(k)$ from Proposition \ref{prop:hoRan}
carries canonically the structure of an $\O_{\ovr{\CC}}\otimes \E_\infty$-algebra.
Hence, $\AAA_\pi\in \QFT_\infty(\ovr{\CC})$ is a homotopy AQFT
in the sense of Definition \ref{def:EinftyhAQFT}.
\end{theo}

\begin{rem}
From a mathematical perspective, the homotopy AQFT
$\AAA_\pi \in \QFT_\infty(\ovr{\CC})$ from Theorem \ref{theo:homotopyinvariants}
can be interpreted in terms of fiber-wise normalized cochain algebras on
$\pi: \ovr{\CC_F}\to\ovr{\CC}$ with coefficients in the strict $\Ch(k)$-valued
AQFT $\AAA\in\QFT(\ovr{\CC_F})$. In other words, it is the fiber-wise
groupoid cohomology of $\pi: \ovr{\CC_F}\to\ovr{\CC}$ with coefficients 
$\AAA\in\QFT(\ovr{\CC_F})$. Similarly to ordinary groupoid cohomology,
the results of this construction can be interesting even when the coefficients
are concentrated in degree $0$.  Physical examples of this type
have been discussed in  \cite{Benini:2016jfs} and they include e.g.\ 
the case of Dirac fields on the groupoid of all possible spin structures over a space-time.
\end{rem}


\appendix

\section{\label{app:cosheaf}On the cosheaf condition in AQFT}
In this appendix we shall analyze an analogue of the 
cosheaf condition \eqref{eqn:cosheafAQFT} for a simple toy-model 
of an AQFT. Our main message will be that it is very hard to find covers
for which this condition holds true. This motivates and justifies
our alternative descent condition that we have 
sketched at the end of Section \ref{sec:background} and 
stated precisely in Definition \ref{def:jlocal}.

The toy-model we consider is given by the scalar field on the circle 
$\bbS^1$, which is {\em not} a Lorentzian AQFT in the sense of 
Definition \ref{def:LCQFT}, but rather a chiral conformal AQFT
on the compactified light ray. (In particular, it is an AQFT in the general
sense of Definition \ref{def:QFTcats}.) Denoting by
$\Open(\bbS^1)$ the category of all open subsets of the circle
$\bbS^1$, our model is described by a functor
$\AAA : \Open(\bbS^1)\to \astAlg$ to the category of $\ast$-algebras.
To an open subset $U\subseteq \bbS^1$, it assigns
the $\ast$-algebra $\AAA(U)$ presented by the following 
generators and relations:
\begin{enumerate}[i)]
\item {\em Generators:} $\Phi_U(\varphi)$, for all compactly supported functions $\varphi\in C^\infty_c(U)$;
\item {\em Relations:}
\begin{enumerate}[i)]
\item {\em $\bbR$-linearity:} $\Phi_U(\alpha\varphi + \beta\psi) = \alpha \Phi_U(\varphi) + \beta\Phi_U(\psi)$,
for all $\alpha,\beta\in\bbR$ and $\varphi,\psi\in C^\infty_c(U)$;
\item {\em Hermiticity:} $\Phi_U(\varphi)^\ast = \Phi_U(\varphi)$, for all $\varphi\in C^\infty_c(U)$;
\item {\em CCR:} $\big[\Phi_U(\varphi),\Phi_U(\psi)\big] = 
\mathrm{i} \,\int_{\bbS^1} \varphi\dd\psi\,\oone$,
 for all  $\varphi,\psi\in C^\infty_c(U)$.
\end{enumerate}
\end{enumerate}
To an open subset inclusion $\iota_U^V : U \subseteq V$,
the functor assigns the $\ast$-homomorphism $\AAA(\iota_U^V):\AAA(U)\to\AAA(V)$
that is defined on the generators by $\Phi_U(\varphi)\mapsto \Phi_V(\varphi)$.
We note that $\AAA$ satisfies the following variant of a causality condition 
(or $\perp$-commutativity condition in the sense of Definition \ref{def:QFTcats});
For every pair of {\em disjoint} open subsets $U_1,U_2\subseteq V$ of 
some open subset $V\subseteq\bbS^1$, the induced commutator
\begin{multline}
\big[\AAA(\iota_{U_1}^V)\big(\Phi_{U_1}(\varphi)\big) ,\AAA(\iota_{U_2}^V) \big(\Phi_{U_2}(\psi)\big)\big]_{\AAA(V)}\\
= \big[\Phi_{V}(\varphi) ,\Phi_{V}(\psi)\big]_{\AAA(V)} = \mathrm{i}\,\int_{\bbS^1}\varphi\dd\psi\,\oone =0
\end{multline}
is zero, for all $\varphi\in C^\infty_c(U_1)$ and $\psi\in C^\infty_c(U_2)$,
because the integrand is zero as a consequence of $U_1\cap U_2=\emptyset$.

For any open cover $\{U_i\subseteq \bbS^1\}$ of the circle,
we define the $\ast$-algebra
\begin{flalign}
\colim\,\AAA(U_\bullet)\,:=\, \colim
\Bigg(
\!\!\xymatrix@C=1em{
\coprod\limits_{ij} \AAA(U_{ij}) \ar@<0.5ex>[r]\ar@<-0.5ex>[r]&
 \coprod\limits_{i} \AAA(U_i) }\!\!
\Bigg)\quad,
\end{flalign}
which also admits a simple presentation by generators and relations:
\begin{enumerate}[i)]
\item {\em Generators:} $\Phi_{U_i}(\varphi)$, for all $i$ and all  $\varphi\in C^\infty_c(U_i)$;
\item {\em Relations:} 
\begin{enumerate}[i)]
\item {\em $\bbR$-linearity:} $\Phi_{U_i}(\alpha\varphi + \beta\psi) = \alpha \Phi_{U_i}(\varphi) + \beta\Phi_{U_i}(\psi)$,
for all $i$, all $\alpha,\beta\in\bbR$ and all $\varphi,\psi\in C^\infty_c(U_i)$;
\item {\em Hermiticity:} $\Phi_{U_i}(\varphi)^\ast = \Phi_{U_i}(\varphi)$, for all $i$ and all $\varphi\in C^\infty_c(U_i)$;
\item {\em CCR:} $\big[\Phi_{U_i}(\varphi),\Phi_{U_i}(\psi)\big] = 
\mathrm{i}\,\int_{\bbS^1} \varphi\dd\psi\,\oone$,
for all $i$ and all  $\varphi,\psi\in C^\infty_c(U_i)$;
\item {\em Overlap relations:} $\Phi_{U_i}(\varphi) = \Phi_{U_j}(\varphi)$, for all $i,j$
and all $\varphi \in C^\infty_c(U_{ij})$.
\end{enumerate}
\end{enumerate}
Note that there are no a priori commutation relations between
$\Phi_{U_i}(\varphi)$ and $\Phi_{U_j}(\psi)$ for different $i\neq j$.
Depending on the cover, there however exist certain induced commutation 
relations that result by combining the CCRs for individual $i$'s 
and the overlap relations. 

There exists a canonical $\ast$-homomorphism
\begin{flalign}\label{eqn:colimmapCFT}
\colim\,\AAA(U_\bullet)\longrightarrow \AAA(\bbS^1)~,~~\Phi_{U_i}(\varphi) \longmapsto \Phi_{\bbS^1}(\varphi)
\end{flalign}
to the $\ast$-algebra on the full circle. We would like to answer the question
for which covers $\{U_i\subseteq\bbS^1\}$ this is an isomorphism, i.e.\ for which covers
the cosheaf condition holds true for our example.
\begin{propo}
\eqref{eqn:colimmapCFT} is an isomorphism if and only if
the open cover $\{ U_i\subseteq \bbS^1\}$ 
satisfies the condition
\begin{flalign}\label{eqn:covercondition}
\forall i,j \,\,\exists k \,:\, U_i\cup U_j \subseteq U_k\quad. 
\end{flalign}
\end{propo}
\begin{proof}
Choosing a partition of unity $\sum_i \chi_i =1$ subordinate to 
$\{ U_i\subseteq \bbS^1\}$, we define for each 
$\varphi\in C^\infty_c(\bbS^1)$ an element
\begin{flalign}\label{eqn:Phitilde}
\widetilde{\Phi}(\varphi) \,:=\, \sum_i \Phi_{U_i}(\chi_i\varphi)~\in \colim\,\AAA(U_\bullet)\quad.
\end{flalign}
These elements are independent of the choice of partition of unity:
For any other choice $\sum_i \rho_i =1$, we obtain
\begin{flalign}
\nonumber \sum_i \Phi_{U_i}(\rho_i\varphi) &= \sum_{i,j}\Phi_{U_{i}}(\rho_i\chi_j \varphi)=
\sum_{i,j}\Phi_{U_{j}}(\rho_i\chi_j \varphi)\\
&= \sum_j \Phi_{U_j}(\chi_j\varphi)=
\widetilde{\Phi}(\varphi) \quad,
\end{flalign}
where in the second step we used the overlap relations for 
$\rho_i\chi_j \varphi\in C^\infty_c(U_{ij})$. The elements in \eqref{eqn:Phitilde}
are clearly $\bbR$-linear in $\varphi\in C^\infty_c(\bbS^1)$ and Hermitian.
Moreover, we have that $\widetilde{\Phi}(\varphi)\mapsto \Phi_{\bbS^1}(\varphi)$
under the map \eqref{eqn:colimmapCFT}. It follows that \eqref{eqn:colimmapCFT}
is an isomorphism, with inverse $\Phi_{\bbS^1}(\varphi) \mapsto \widetilde{\Phi}(\varphi)$,
if and only if the $\widetilde{\Phi}(\varphi)$'s, for $\varphi \in C^\infty_c(\mathbb{S}^1)$, satisfy the CCR. These are equivalent to the 
commutation relations
\begin{flalign}
\big[\Phi_{U_i}(\varphi), \Phi_{U_j} (\psi)\big] \,=\, \mathrm{i} \int_{\bbS^1} \varphi\dd\psi~\oone
\end{flalign}
in $\colim\,\AAA(U_\bullet)$, for all $i,j$ and all $\varphi\in C^\infty_c(U_i)$ and $\psi\in C^\infty_c(U_j)$,
which are satisfied if and only if condition \eqref{eqn:covercondition} holds.
\end{proof}

\begin{lem}
An open cover $\{ U_i\subseteq \bbS^1\}$ satisfies \eqref{eqn:covercondition}
if and only if one of its members is the whole circle $\bbS^1$.
\end{lem}
\begin{proof}
The direction `$\Leftarrow$' is obvious. To prove `$\Rightarrow$',
note that each open cover of the compact space $\bbS^1$ has a finite subcover, say 
$U_1,\dots,U_N\subseteq \bbS^1$. Applying \eqref{eqn:covercondition} iteratively, 
we obtain
$\bbS^1 =  U_1\cup U_2 \cup \cdots \cup U_N \subseteq U_i \cup U_3 \cup\cdots\cup 
U_N\subseteq \cdots\subseteq U_k$, i.e.\ there exists $k$ such that $U_k=\bbS^1$ is the circle.
\end{proof}

\begin{cor}
The AQFT $\AAA : \Open(\bbS^1)\to \astAlg$ describing the
scalar field on $\bbS^1$ satisfies the cosheaf condition 
for an open cover $\{ U_i\subseteq \bbS^1\}$ 
if and only if $\bbS^1$ is a member of this cover.
In particular, it does {\em not} satisfy the cosheaf condition
for any open cover $\{I_i \subseteq\bbS^1\}$ by intervals.
\end{cor}

\begin{rem}
Clearly, a cosheaf condition for covers containing the entire space 
holds trivially and hence has no power. From the discussion above 
we deduce that even for an elementary prototypical example of AQFT, such 
as the scalar field on the circle $\bbS^1$, the cosheaf condition 
holds only for covers that contain the circle itself. 
A similar behavior arises more generally also in AQFTs on Lorentzian manifolds. 
The alternative descent condition {\em $j$-locality}, 
inspired by Fredenhagen's universal algebra construction \cite{Fredenhagen:1989pw,Fredenhagen:1993tx,Fredenhagen:1992yz} 
and formalized in Definition \ref{def:jlocal}, is better behaved 
in standard examples of AQFTs, cf.\ Example \ref{ex:jlocal}.
\end{rem}


\bibliography{allbibtex}

\bibliographystyle{prop2015}

\end{document}